\documentclass[twoside,11pt]{article}
\pdfoutput=1
\usepackage{afterpage}
\usepackage{subfigure}
\usepackage{algorithm}
\usepackage{algorithmic}
\usepackage{multirow}
\usepackage{bibentry}
\usepackage{amsmath}
\usepackage{amssymb}
\usepackage{amsfonts}
\usepackage{amsthm}
\usepackage{nicefrac}
\usepackage{epsfig}
\usepackage{epstopdf}
\usepackage{graphicx}
\usepackage{tabularx}
\usepackage{enumitem}
\usepackage{url}
\setlist{nolistsep}
\usepackage{indentfirst}
\usepackage{subfigure}
\usepackage{comment}
\usepackage{color}
\usepackage{verbatim}
\usepackage{bm}

\newcommand{\argmin}{\operatornamewithlimits{argmin}}

\newcommand{\circR}{\operatornamewithlimits{circ}}
\newcommand{\sign}{\operatornamewithlimits{sign}}
\newcommand{\Tr}{\operatornamewithlimits{Tr}}

\newcommand{\bR}{\mathbf{R}}
\newcommand{\br}{\mathbf{r}}
\newcommand{\bD}{\mathbf{D}}

\newcommand{\bA}{\mathbf{A}}
\newcommand{\bx}{\mathbf{x}}
\newcommand{\by}{\mathbf{y}}
\newcommand{\bu}{\mathbf{u}}
\newcommand{\bv}{\mathbf{v}}
\newcommand{\bz}{\mathbf{z}}

\newcommand{\bb}{\mathbf{b}}
\newcommand{\ba}{\mathbf{a}}
\newcommand{\be}{\mathbf{e}}
\newcommand{\Bf}{\mathbf{f}}
\newcommand{\bp}{\mathbf{p}}
\newcommand{\bq}{\mathbf{q}}

\newcommand{\calN}{\mathcal{N}}
\newcommand{\calS}{\mathcal{S}}
\newcommand{\E}{\mathbb{E}}
\newcommand{\real}{\mathbb{R}}
\newcommand{\norm}[1]{\lVert #1 \rVert}

\newcommand{\spn}[1]{\text{span}\{ #1 \}}
\newcommand{\iprod}[1]{\langle #1 \rangle}
\renewcommand{\[}{\begin{equation}}
\renewcommand{\]}{\end{equation}}

\newcommand{\circu}[1]{\mathcal{C}_{#1}}

\graphicspath{{Figure/}}
\usepackage{jmlr2e_arXiv}

\jmlrheading{1}{2015}{1-48}{4/00}{10/00}{F. Yu, A. Bhaskara, S. Kumar,  Y. Gong,  S.-F. Chang}
\ShortHeadings{On Binary Embedding using Circulant Matrices}{Yu, Bhaskara, Gong, Kumar and Chang}
\firstpageno{1}

\begin{document}

\title{On Binary Embedding using Circulant Matrices}

\author{\name Felix X. Yu$^{1,2}$ \email felixyu@google.com \\
        \name Aditya Bhaskara$^{1}$ \email bhaskaraaditya@gmail.com \\
        \name Sanjiv Kumar$^{1}$ \email sanjivk@google.com \\
        \name Yunchao Gong$^{3}$ \email yunchao@cs.unc.edu  \\        
        \name Shih-Fu Chang$^{2}$ \email sfchang@ee.columbia.edu \\
        Google Research, New York, NY 10011\\
        Columbia University, New York, NY 10027\\
		Snapchat, Inc., Venice, CA 90291
        }
\editor{}

\maketitle

\begin{abstract}
Binary embeddings provide efficient  and powerful ways to perform operations on large scale data.  However binary embedding typically requires long codes in order to preserve the discriminative power of the input space. Thus binary coding methods traditionally suffer from high computation and storage costs in such a scenario. To address this problem, we propose Circulant Binary Embedding (CBE) which generates binary codes by projecting the data with a circulant matrix. The circulant structure allows us to use Fast Fourier Transform algorithms to speed up the computation. For obtaining $k$-bit binary codes from $d$-dimensional data, our method improves the 
time complexity from $\mathcal{O}(dk)$ to $\mathcal{O}(d\log{d})$, and the space complexity from $\mathcal{O}(dk)$ to $\mathcal{O}(d)$.

We study two settings, which differ in the way we choose the parameters of the circulant matrix.  In the first, the parameters are chosen randomly and in the second, the parameters are learned using the data. For randomized CBE, we give a theoretical analysis comparing it with binary embedding using an unstructured random projection matrix.  The challenge here is to show that the dependencies in the entries of the circulant matrix do not lead to a loss in performance.  In the second setting, we design a novel time-frequency alternating optimization to learn data-dependent circulant projections, which alternatively minimizes the objective in original and Fourier domains. In both the settings, we show by extensive experiments that the CBE approach gives much better performance than the state-of-the-art approaches if we fix a running time, and provides much faster computation with negligible performance degradation if we fix the number of bits in the embedding.

\end{abstract}
\begin{keywords}
Circulant Matrix, Dimension Reduction, Binary Embedding \\ ~\\
{\bf Note.} A preliminary version of this article with the first, third, fourth and fifth authors appeared in the Proceedings of ICML 2014.
\end{keywords}

\section{Introduction}\label{sec:intro}
Sketching and dimensionality reduction have become powerful and ubiquitous
tools in the analysis of large high-dimensional datasets, with
applications ranging from computer vision, to biology, to finance.
The celebrated Johnson-Lindenstrauss lemma says that projecting high
dimensional points to a random $\mathcal{O}(\log N)$-dimensional space
approximately preserves all the pairwise distances between a set of
$N$ points, making it a powerful tool for nearest neighbor search,
clustering, etc.  This started the paradigm of designing low
dimensional {\em sketches} (or embeddings) of high dimensional data
that can be used for efficiently solving various information retrieval
problems.

More recently, {\em binary} embeddings (or embeddings into $\{0,1\}^k$
or $\{-1,1\}^k$) have been developed for problems in which we care
about preserving the {\em angles} between high dimensional
vectors~\cite{li2011hashing, gonglearning, raginsky2009locality,
gong2012angular, liu2011hashing}.  The main appeal of binary
embeddings stems from the fact that storing them is often much more
efficient than storing real valued embeddings. Furthermore, operations
such as computing the Hamming distance in binary space can be
performed very efficiently either using table lookup, or
hardware-implemented instructions on modern computer architectures.

In this paper, we study binary embeddings of high-dimensional data.
Our goal is to address one of its main challenges:
even though binary embeddings are easy to manipulate, it
has been observed that obtaining high accuracy results {\em
requires} the embeddings to be rather long when the data is high
dimensional~\cite{li2011hashing, gonglearning, sanchez2011high}. Thus in applications like computer
vision, biology and finance (where high dimensional data is common),
the task of computing the embedding is a bottleneck.  The natural algorithms have time and space
complexity $\mathcal{O}(dk)$ per input point in order to produce a $k$-bit
embedding from a $d$-dimensional input.  Our main contribution in this work is to
improve these complexities to $\mathcal{O}(d \log d)$ for time and $\mathcal{O}(d)$ for space
complexity.

Our results can be viewed as binary analogs of the recent work on {\em
fast} Johnson-Lindenstrauss transform.  Starting with the work of
Ailon and Chazelle~\cite{ailon2006approximate}, there has been a lot
of beautiful work on fast algorithms for dimension reduction with the
goal of preserving pairwise distances between points.  Various
aspects, such as exploiting sparsity, and using structured matrices to
reduce the space and time complexity of dimension reduction, have been
explored~\cite{ailon2006approximate, matouvsek2008variants, liberty2008dense}.  
But the key difference in our setting is
that binary embeddings are {\em non-linear}.  This makes the analysis
tricky when the projection matrices do not have independent entries.
Binary embeddings are also better suited to approximate the angles
between vectors (as opposed to distances). Let us see why.

The general way to compute a binary embedding of a data point
$\bx \in \real^d$ is to first apply a linear transformation $\bA \bx$
(for a $k \times d$ matrix $\bA$), and then apply a {\em binarization}
step.  We consider the natural binarization of taking the sign.  Thus,
for a point $\bx$, the binary embedding into $\{-1,1\}^d$ we consider
is
\begin{equation}
h(\mathbf{x}) = \text{sign} (\bA \mathbf{x}),
\label{eq:lsh_intro}
\end{equation}
where $\bA \in \real^{k \times d}$ as above, and $\text{sign}(\cdot)$ is
a binary map which returns element-wise sign\footnote{A few methods
transform the linear projection via a nonlinear map before taking the
sign \cite{weiss2008spectral, raginsky2009locality}.}. How should one
pick the matrix $\bA$? One natural choice, in light of the
Johnson-Lindenstrauss lemma, is to pick it randomly, i.e., each entry
is sampled from an independent Gaussian. This {\em data oblivious} choice is
well studied~\cite{charikar2002similarity, raginsky2009locality}, and
has the nice property that for two data vectors $\bx, \by$, the
$\ell_1$ distance between their embeddings is proportional to the
angle between $\bx$ and $\by$, in expectation (over the random entries
in $\bA$). This is a consequence of the fact that for any
$\bx, \by \in \real^d$, if $\br$ is drawn from $\calN(0,1)^d$,
\[\Pr[ \sign{\iprod{\bx, \br}} = \sign{\iprod{\by, \br}}] = \frac{\angle(\bx, \by)}{\pi}.     \]

Other data oblivious methods have also been studied in the literature,
by choosing different distributions for the entries of $\bA$. While
these methods do reasonably well in practice, the natural question is
if adapting the matrix to the data allows us to use shorter codes
(i.e., have a smaller $k$) while achieving a similar error. A number
of such data-dependent techniques have been proposed with different
optimization criteria such as reconstruction
error \cite{Kulislearningto}, data dissimilarity \cite{Norouzi11,
weiss2008spectral}, ranking loss \cite{norouzi2012Hamming},
quantization error after PCA \cite{gongiterative}, and pairwise
misclassification \cite{wang2010sequential}. As long as data is
relatively low dimensional, these methods have been shown to be quite
effective for learning compact codes.
 
However, the $\mathcal{O}(kd)$ barrier on the space and time complexity barrier
prevents them from being applied with very high-dimensional data. For
instance, to generate 10K-bit binary codes for data with 1M
dimensions, a huge projection matrix will be required needing tens of
GB of memory.\footnote{In the oblivious case,
one can generate the random entries of the matrix on-the-fly (with
fixed seeds) without needing to store the matrix, but this increases
the computational time even further.}

In order to overcome these computational
challenges, \cite{gonglearning} proposed a {\em bilinear projection}
based coding method. The main idea here is to reshape the input vector
$\mathbf{x}$ into a matrix $\mathbf{Z}$, and apply a bilinear
projection to get the binary code:
\begin{eqnarray}
h(\mathbf{x}) = \text{sign} (\mathbf{R}_1^T \mathbf{Z} \mathbf{R}_2).
\end{eqnarray}
When the shapes of $\mathbf{Z}, \mathbf{R}_1, \mathbf{R}_2$ are chosen
appropriately\footnote{Specifically,
$\mathbf{Z} \in \mathbb{R}^{\sqrt{d} \times \sqrt{d}}$,
$\mathbf{R}_1, \mathbf{R}_2 \in \mathbb{R}^{\sqrt{k} \times \sqrt{d}}$.},
the method has time and space complexities $\mathcal{O}(d\sqrt{k})$
and $\mathcal{O}(\sqrt{dk})$ respectively.  Bilinear codes make it
feasible to work with datasets of very high dimensionality and have
shown good results for a variety of tasks.

\subsection{Our results}
In this work, we propose a novel technique, called Circulant Binary
Embedding (CBE), which is even faster than the bilinear coding. The
main idea is to impose a {\em circulant} (described in detail in
Section~\ref{sec:cbe}) structure on the projection matrix $\bA$ in
(\ref{eq:lsh_intro}).  This special structure allows us to compute the
product $\bA\bx$ in time $\mathcal{O}(d \log d)$ using the Fast Fourier Transform
(FFT), a tool of great significance in signal processing. The space
complexity is also just $\mathcal{O}(d)$, making it efficient even for very high
dimensional data.
Table \ref{table:methods} compares the time and space complexity for
the various methods outlined above.

Given the efficiency of computing the CBE, two natural questions
arise: how good is the obtained embedding for various information
retrieval tasks? and how should we pick the parameters of the
circulant $\bA$?
 
In Section~\ref{sec:rand}, we study the first question for {\em
random} CBE, i.e., when the parameters of the circulant are picked
randomly (independent Gaussian, followed by its shifts). Specifically,
we analyze the {\em angle estimating} property of binary embeddings
(Eq.\eqref{eq:lsh_intro}), which is the basis for its use in
applications. Under mild assumptions, we show that using a random
circulant $\bA$ has the same qualitative guarantees as using fully
random $\bA$. These results provide some of the few theoretical guarantees we
are aware of, for non-linear circulant-based embeddings. We defer the
formal statements of our results to Section~\ref{sec:rand},
Theorems~\ref{thm:variance-cbe} and~\ref{thm:jl-large-angle}. We note that in independent and very recent work, Choromanska et al.~\cite{Choromanska15} obtain a qualitatively similar analysis of CBE, however the bounds are incomparable to ours.

In Section~\ref{sec:opt}, we study the second question, i.e., learning
data-dependent circulant matrices. We propose a novel and efficient
algorithm, which alternatively optimizes the objective in the original
and frequency domains.

Finally in Section~\ref{sec:exp}, we study the empirical performance
of circulant embeddings via extensive experimentation. Compared to the
state-of-the-art, our methods improve the performance dramatically for
a fixed computation time. If we instead fix the number of bits in the
embedding, we observe that the performance degradation is negligible,
while speeding up the computation many-fold (see
Section~\ref{sec:exp}).

\begin{table}
\begin{center}
\begin{tabular}{l|l|l|l}
\hline Method & Time   & Space  & Time (optimization)\\ 
\hline   Unstructured      & $\mathcal{O}(dk)$ & $\mathcal{O}(dk)$ & $\mathcal{O}(Nd^2k)$\\ 
\hline   Bilinear  & $\mathcal{O}(d\sqrt{k})$ &  $\mathcal{O}(\sqrt{dk})$ &    $\mathcal{O}(Nd\sqrt{k})$\\ 
\hline   Circulant ($k \leq d$) & $\mathcal{O}(d\log{d})$ & $\mathcal{O}(d)$  & $\mathcal{O}(Nd\log{d})$\\ 
\hline   Circulant ($k > d$) & $\mathcal{O}(k\log{d})$ & $\mathcal{O}(k)$  & $\mathcal{O}(Nk\log{d})$\\ 
\hline 
\end{tabular}
\end{center}
\caption{Comparison of the time and space complexities. 
$d$ is the input dimensionality, and $k$ is the output dimensionality
(number of bits). $N$ is the number of instances used for learning
data-dependent projection matrices.  See Section \ref{subsec:kned} for
discussions on $k < d$ and $k > d$.  }
\label{table:methods}
\end{table}

\section{Background and related work}\label{sec:related}
The lemma of Johnson and Lindenstrauss~\cite{johnson1984extensions} is
a fundamental tool in the area of sketching and dimension
reduction. The lemma states that if we have $N$ points in
$d$-dimensional space, projecting them to an $\mathcal{O}(\log N)$ dimensional
space (independent of $d$!) preserves all pairwise distances.
Formally,

\begin{lemma}[Johnson Lindenstrass lemma]
\label{lemma:jl_construct}
Let $S$ be a set of $N$ points in $\mathbb{R}^d$. Let
$\mathbf{A} \in \mathbb{R}^{k \times d}$ be a matrix whose entries are
drawn i.i.d from $\mathcal{N}(0,1)$.  Then with probability at least
$1 - 2N^2 e^{-(\epsilon^2-\epsilon^3 ) k / 4}$
\[
	(1 - \epsilon) \norm{\bx - \by}_2 \leq
	\frac{1}{\sqrt{k}} \norm{\mathbf{A}(\bx - \by)}_2 
	\leq (1 + \epsilon) \norm{\bx - \by}_2
\]
for any $\bx, \by \in S$.
\end{lemma}
When $k = \mathcal{O}(\log N/\epsilon^2)$, the probability above can be made
arbitrarily close to $1$. Due to the simplicity and theoretical
support, random projection based dimensionality reduction has been
applied in broad applications including approximate nearest neighbor
research \cite{indyk1998approximate}, dimensionality reduction in
databases \cite{achlioptas2003database}, and bi-Lipschitz embeddings
of graphs into normed spaces \cite{frankl1988johnson}.

However a serious concern in a few applications is the dependence of
$k$ on the accuracy ($\mathcal{O}(1/\epsilon^2)$). The space and time complexity
of dimension reduction are $\mathcal{O}(kd)$, if the computation is done in the
natural way. Are there faster methods when $k$ is reasonably large?
As mentioned earlier, the line of work starting
with~\cite{ailon2006approximate} aims to improve the time and space
complexity of dimension reduction.  This led to work showing
Johnson-Lindenstruss-type guarantees with {\em structured} matrices
(with some randomness), including Hadamard matrices along with a sparse
random Gaussian matrix \cite{ailon2006approximate}, sparse
matrices \cite{matouvsek2008variants}, and Lean Walsh
Transformations \cite{liberty2008dense}. The advantage of using
structured matrices is that the space and computation cost can be
dramatically reduced, yet the distance preserving property remains to
be competitive.

In this context, randomized circulant matrices (which are also the
main tool in our work) have been studied, starting with the
works~\cite{hinrichs2011johnson, vybiral2011variant}.  The dimension
reduction comprises of random sign flips followed by multiplication by
a randomized circulant matrix. For $d$-dimensional input, reducing the
dimension to $k$ for $k<d$ has time complexity $\mathcal{O}(d \log d)$ and space
complexity $\mathcal{O}(d)$, independent of $k$.
Proving bounds similar to Lemma~\ref{lemma:jl_construct} turns out to
be much more challenging because the entries of the projection matrix
are now highly dependent, and thus concentration bounds are hard to
prove.
The first analysis~\cite{hinrichs2011johnson} showed that reducing to
$\mathcal{O}(\log^3 N/\epsilon^2)$ dimensions (compared to $\mathcal{O}(\log
N/\epsilon^2)$ in Lemma~\ref{lemma:jl_construct}) preserves all
pairwise distances with high probability. This was improved to
$\mathcal{O}(\log^2 N/\epsilon^2)$ in \cite{vybiral2011variant}, and furthermore
to $\mathcal{O}(\log^{(1 + \delta)} N/\epsilon^2)$ in \cite{zhang2013new}, using
matrix-valued Bernstein inequalities. These works provide the
motivation for our theoretical results, however the key difference for
us is the {\em binarization} step, which is highly non-linear. Thus we
need to develop new machinery for our analysis.

\paragraph{Binary embeddings.}
Recently, structured matrices used in the context of the fast JL
transform (a combination of Hadamard and sparse random Gaussian
matrices) have also been studied for binary
embedding~\cite{dasgupta2011fast}, and more
recently~\cite{yi2015binary}. In particular, \cite{yi2015binary}
showed that the method can achieve $\epsilon$ distance preserving
error with $\mathcal{O}(\log N/\epsilon^2)$ bits and $\mathcal{O}(d \log d)$ computational
complexity, for $N$ points ($N \ll \epsilon \sqrt{d}$).
In this work, we study the application of using the circulant matrix
for binary embedding. The work extends and provides theoretical
justification for our previous conference paper on this
topic \cite{icml14_cbe}.

The idea of using structured matrices to speed up linear projection
has also be exploited under the settings of deep neural
networks \cite{arxiv_circulant_nn, yang2014deep}, and kernel
approximation \cite{arxiv_cnm, le2013fastfood} .

\section{Circulant Binary Embedding}\label{sec:cbe}
Let us start by describing our framework and setting up the notation
that we use in the rest of the paper.

\subsection{The Framework}
\label{subsec:framework}
We will now describe our algorithm for generating $k$-bit binary codes
from $d$-dimensional real vectors. We start by discussing the case
$k=d$ and move to the general case in Section~\ref{subsec:kned}.  The
key player is the circulant matrix, which is defined by a real vector
%
%
$\br = (r_0, r_1, \cdots, r_{d-1})^T$ \cite{gray2006toeplitz}.
\begin{align}
\circu{\br} :=
\begin{bmatrix}
r_0     & r_{d-1} & \dots  & r_{2} & r_{1}  \\
r_{1} & r_0    & r_{d-1} &         & r_{2}  \\
\vdots  & r_{1}& r_0    & \ddots  & \vdots   \\
r_{d-2}  &        & \ddots & \ddots  & r_{d-1}   \\
r_{d-1}  & r_{d-2} & \dots  & r_{1} & r_{0}
\end{bmatrix}.
\label{eq:cir}
\end{align}

Let $\mathbf{D}$ be a diagonal matrix with each diagonal entry
$\sigma_i$, $i = 0, \cdots, d-1$, being a Rademacher variable ($\pm 1$
with probability 1/2):
\begin{align}
\mathbf{D} =
\begin{bmatrix}
\sigma_0     &  &   &  &   \\
 & \sigma_1    &  &    &   \\
  & & \sigma_2    &   &    \\
  &  &  &  \ddots &    \\
  & &   &  & \sigma_{d-1}
\end{bmatrix}.
\end{align}
For $\mathbf{x} \in \mathbb{R}^d$, its $d$-bit Circulant Binary
Embedding (CBE) with $\mathbf{r} \in \mathbb{R}^{d}$ is defined as:
\begin{align}
h(\mathbf{x}) = \text{sign} (\circu{\br} \mathbf{D}  \mathbf{x}),
\label{eq:def}
\end{align}
where $\circu{\br}$ is defined as above. Note that applying
$\mathbf{D}$ to $\mathbf{x}$ is equivalent to applying a random sign
flip to each coordinate of $\mathbf{x}$. The necessity of such an
operation is discussed in the introduction of Section \ref{sec:rand}.
Since sign flipping can be carried out as a preprocessing step for
each input $\mathbf{x}$, here onwards for simplicity we will drop
explicit mention of $\mathbf{D}$. Hence the binary code is given as
$h(\mathbf{x}) = \text{sign} (\circu{\br} \mathbf{x})$.

\subsection{Computational Complexity}
\label{subsec:space_time}
The main advantage of a circulant based embedding is that it can be computed
quickly using the Fast Fourier Transform (FFT). The following is a
folklore result, whose proof we include for completeness.

\begin{proposition}
For a $d$-dimensional vector $\bx$ and any $\br \in \Re^d$, the $d$-bit CBE
$\sign (\mathcal{C}_r (\bD \bx))$ can be computed using $\mathcal{O}(d)$ space and $\mathcal{O}(d \log d)$
time.\label{prop:time}
\end{proposition}

\begin{proof}
The space complexity comes only from the storage of the vector $\br$
and the signs $\bD$ (which amount to $\mathcal{O}(d)$). We never need to store
the full matrix $\circu{\br}$ explicitly.

The main property of a circulant matrix is that for any vector $\by \in \bR^d$,
we can compute $\circu{\br}\by$ in time $\mathcal{O}(d \log d)$.  This is because
\begin{equation} \circu{\br} = \mathcal{F}_d^{-1} ~\text{diag}(\mathcal{F}_d \br)~ \mathcal{F}_d, \label{eq:circ-diagonalization}
\end{equation}
where $\mathcal{F}_d$ is the matrix corresponding to the Discrete Fourier Transform (DFT)
of periodicity $N$, i.e., whose $(i,j)$th entry is given by
\[ \mathcal{F}_d (i,j) = \omega^{ij}, \]
where $\omega$ is the $N$th root of unity $e^{-2\pi \iota/N}$.  The celebrated Fast Fourier Transform algorithm~\cite{oppenheim1999discrete} says that for any $\bz \in \bR^d$, we can compute $\mathcal{F}_d \bz$ and $\mathcal{F}_d^{-1} \bz$ in time $\mathcal{O}(d \log d)$, using $\mathcal{O}(d)$ space. This immediately implies that we can compute $\mathcal{C}_\br \by$ within the same space and time complexity bounds.
\end{proof}

\subsection{Generalizing to $k \ne d$}
\label{subsec:kned}
The computation above assumed that number of bits we produce ($k$) is
equal to the input dimension. Let us now consider the general case.

When $k<d$, we still use the circulant matrix $\bR \in \mathbb{R}^{d
  \times d}$ with $d$ parameters, but the output is set to be the
first $k$ elements in (\ref{eq:def}). This is
equivalent to the operation
\[ \Phi(\bx) := \sign( \circu{\br, k} \bD \bx ),\]
where $\circu{\br, k}$ the so-called {\em partial circulant matrix},
which is $\circu{\br}$ truncated to $k$ columns.  We note that CBE with $k < d$ is not computationally more efficient than that with $k=d$.

When $k>d$, using a single $\br$ causes repetition of bits, so we
propose using $\circu{\br}$ for multiple $\br$, and concatenating
their output. This gives the computational complexity $\mathcal{O}(k \log d)$,
and space complexity $\mathcal{O}(k)$.  Note that as the focus of this paper is
on binary embedding on high-dimensional data, from here onwards, we
assume $k \leq d$. The $k > d$ case is useful in other applications
such as neural network \cite{iccv15_circulant} and kernel
approximation \cite{arxiv_cnm}.

\subsection{Choosing the Parameters $\br$}
We have presented the general framework as well as its space and
computation efficiency in this section. One critical question left
unanswered is how to decide the parameter $\br$.  As mentioned in the
introduction, we consider two solutions. In Section \ref{sec:rand}, we
study the randomized version, where each element of $\br$ is
independently sampled from a unit Gaussian distribution. This is
inspired by the popular Locality Sensitive Hashing (simhash) approach.
 Section \ref{sec:opt} introduces an optimized version, where the
parameters are optimized based on training data and an distance
preserving objective function.

\section{Randomized CBE -- A Theoretical Analysis}\label{sec:rand}
We now analyze the angle preserving properties of CBE when the
circulant matrix used is generated from a random $d$-dimensional
vector.  Formally, we consider the partial circulant matrix
$\circu{\br, k}$, for $\br \sim \calN(0,1)^d$. The embedding we
consider for an $\bx \in \real^d$ is given by
\[ \Phi(\bx) := \sign(\circu{\br, k} \bD \bx). \]
As before, $\bD$ is a diagonal matrix of signs. Hence the embedding
uses $2d$ independent `units' of randomness.

Now, for any two vectors $\bx, \by \in \real^d$, we have that
\begin{equation}
\E \left[ \frac{1}{2k} \norm{\Phi(\bx)- \Phi(\by)}_1 \right] = \frac{\angle(\bx, \by)}{\pi},\label{eq:expectation-angle}
\end{equation}
implying that the random variable $(1/2k) \norm{\Phi(\bx)
- \Phi(\by)}_1 $ provides an estimate for $\theta/\pi$, where $\theta
:= \angle(\bx, \by)$.

We present two main results. In the first, we bound the variance of
the above angle estimate for given $\bx, \by$. We compare with the variance in the {\em
fully independent} case, i.e., when we consider the embedding
$\sign(\bA \bx)$, where $\bA$ is a $k \times d$ matrix with all entries
being independent (and unit normal).  In this case, the variance of
the estimator in Eq.~\eqref{eq:expectation-angle} is equal to
$\frac{1}{k} \frac{\theta}{\pi}\left( 1- \frac{\theta}{\pi}\right)
$.\footnote{We are computing the variance of an average of
i.i.d. Bernoulli random variables which take value $1$ with
probability $p = \theta/\pi$.}

We show that using a circulant matrix instead of $A$ above has a
similar dependence on $k$, as long as the vectors are {\em well
spread}.  Formally,

\begin{theorem}
\label{thm:variance-cbe}
Let $\bx, \by \in \mathbb{R}^d$, such that
$\max\{ \norm{\bx}_\infty/\norm{\bx}_2, \norm{\by}_\infty/{\norm{\by}_2} \} \leq \rho$,
for some parameter $\rho < 1$, and set $\theta = \angle(\bx, \by)$.  The
variance of the averaged hamming distance of $k$-bit code generated by
randomized CBE is
\begin{equation}
	var\left[\frac{1}{2k} \norm{\Phi_x - \Phi_y}_1 \right] \leq
\frac{1}{k} \frac{\theta}{\pi} \left(1 - \frac{\theta}{\pi}\right) + 32 \rho.
\end{equation}
The variance above is over the choice of $\br$ and the random signs
$\bD$.
\end{theorem}

\paragraph{Remark.}
For typical vectors in $\mathbb{R}^d$, we have
${\norm{\bx}_\infty}/{\norm{\bx}_2} $ to be $O(\log
d/\sqrt{d})$. Further, by using the idea from Ailon and Chazelle~\cite{ailon2006approximate}, we can pre-process the data by multiplying it with a randomly signed Hadamard matrix, and guarantee such an $\ell_\infty$ bound with high probability.\footnote{However, applying this pre-processing leads to {\em dense} vectors, which may be memory intensive for some applications. In this case, dividing the co-ordinates into blocks of size $\sim k^2$ and performing the pre-processing on the blocks separately is better for small $k$.} Therefore the second term becomes negligible for large
$d$. The above result suggests that the angle preservation performance
of CBE (in term of the variance) is as good as LSH for
high-dimensional data. 

Our second theorem gives a large-deviation bound for the angle
estimate, also assuming that the vectors are well-spread. This will
then enable us to obtain a dimension reduction theorem which preserves
all angles up to an additive error. 

\begin{theorem}\label{thm:jl-large-angle}
Let $\bx, \by \in \bR^d$ with $\angle( \bx, \by ) = \theta$, and
suppose $\max\{ \norm{\bx}_\infty / \norm{\bx}_2, \norm{\by}_\infty
/ \norm{\bx}_2\} \leq \rho$, for some parameter $\rho$.  Now consider
the $k$-dimensional CBE $\Phi_x, \Phi_y$ of $\bx, \by$ respectively,
for some $k < d$.  Suppose $\rho \le \frac{\theta^2}{16
k \log(k/\delta)}$. For any $\epsilon>0$, we have:
\[   \Pr \left[ \left| \frac{1}{2k}\norm{\Phi_x - \Phi_y}_1 - \frac{\theta}{\pi} \right| > \frac{4 \log(k/\delta)}{\sqrt{k}} \right] <  \delta .\]
\end{theorem}

Qualitatively, the condition on $\rho$ is similar to the one we implicitly have in Theorem~\ref{thm:variance-cbe}.  Unless $\rho = o \big( \frac{1}{k} \frac{\theta}{\pi} (1-\frac{\theta}{\pi}) \big)$, the additive term dominates, so for the bound to be interesting, we need this condition on $\rho$. 

We observe that Theorem~\ref{thm:jl-large-angle} implies a Johnson-Lindenstrauss type theorem.  

\begin{corollary}
Suppose we have $N$ vectors $\bu_0, \bu_1, \dots, \bu_{N-1}$ in $\bR^d$, and define
\[ \rho_{ij} = \max\{ {\norm{\bu_i}_\infty}/{\norm{\bu_i}_2}, {\norm{\bu_j}_\infty}/{\norm{\bu_j}_2} \}, \quad \theta_{ij} = \angle(\bu_i, \bu_j). \]
Let $\epsilon >0$ be a given accuracy parameter and let $k = C\log^2 n/ \epsilon^2$.  Then for all $i, j$ such that $\rho_{ij} < \frac{\theta_{ij}^2}{16 k \log(2kN^2)}$, we have
\[ \left| \frac{1}{2k} \norm{\Phi_i - \Phi_j}_1 - \frac{\theta_{ij}}{\pi} \right| < \epsilon,\]
with probability at least $3/4$.
\end{corollary}
\begin{proof}
We can set $\delta = 1/2N^2$ in Theorem~\ref{thm:jl-large-angle} and then take a union bound over all $\binom{N}{2}$ choices of pairs $i,j$ to obtain a failure probability $\le 1/4$.  Further, for our choice of $k$, setting $C = 144$ and assuming $N$ is large enough that $k < N$, we have
\[ 
\frac{ 4 \log(k/\epsilon)}{\sqrt{k}}  < \frac{ 12\delta  \log N}{ \sqrt{C} \cdot \log N} < \epsilon.
\]
\end{proof}

In the remainder of the section, we will prove the above theorems.  We start with Theorem \ref{thm:variance-cbe}, whose proof will give a basic framework for that of Theorem~\ref{thm:jl-large-angle}.

\subsection{Variance of the angle estimator}
For a vector $\bx$ and an index $i$, we denote by $s_{\rightarrow i} (\bx)$ the vector {\em shifted} by $i$ positions. I.e., the $j$th entry of $s_{\rightarrow i}$ is the $((j-i)\text{mod } d)$'th entry of $\bx$. Further, let us define
\[
F_i = \frac{1 - \sign(s_{\rightarrow i}(\br)^T \mathbf{D} \bx
) \sign(s_{\rightarrow i}(\br)^T \mathbf{D} \by ) }{2}
- \frac{\theta}{\pi}.
\]
where $s_{\rightarrow i}(\cdot)$ is defined as the operator circularly
shifting a vector by $i$ elements\footnote{The above comes with a
slight abuse of notation, where the first column (instead of row) of
the projection matrix $\bR$ is defined as $\br$.}.  We have
\[
var\left[\frac{1}{2k} \norm{\Phi_x - \Phi_y}_1 \right] =
var\left[ \frac{1}{k} \sum_{i=1}^k F_i \right].
\]
Without loss of generality, we assume $\norm{\bx}_2, \norm{\by}_2 = 1$
(since we only care about the angle). The mean of each $F_i$ is zero, and thus $\E [ \frac{1}{k} \sum_{i=1}^k F_i] = 0$.  Thus the variance is equal to
\begin{align}
var\left[\frac{1}{k} \sum_{i=0}^{k-1} F_i \right]
    &= \mathbb{E}\left[ \frac{1}{k^2} \left( \sum_{i=0}^{k-1}
    F_i \right)^2 \right] \\
    &= \mathbb{E} \left[ \frac{\sum_{i=0}^{k-1}
    F_i^2 + \sum_{i \ne j} F_i F_j } {k^2} \right] \notag \\
    &= \frac{1}{k^2} \left(k \cdot \mathbb{E}F_1^2 + \sum_{i \ne
    j} \mathbb{E}(F_i F_j) \right) \notag \\
    &= \frac{1}{k} \frac{\theta}{\pi} \left(1
    - \frac{\theta}{\pi}\right) + \frac{1}{k^2} \sum_{i \ne
    j} \mathbb{E}(F_i F_j) \label{eq:to-show-simplified}
\end{align}

To prove the theorem, it suffices to show that $\mathbb{E}(F_i F_j)
\leq 32\rho$ for all $i \ne j$.
Without loss of generality, we can assume that $i=0$, and consider
$\E(F_0 F_j)$.  By definition, it is equal to
\[
\mathbb{E} \left[\left(\frac{1 - \sign(\br^T \mathbf{D} \bx ) \sign(\br^T \mathbf{D} \by ) }{2} - \frac{\theta}{\pi}  \right)
    \left(\frac{1 - \sign(s_{\rightarrow j}(\br)^T \mathbf{D} \bx
    ) \sign(s_{\rightarrow j}(\br)^T \mathbf{D} \by ) }{2}
    - \frac{\theta}{\pi} \right) \right]. \nonumber
\]
The trick now is to observe that
\begin{equation}
s_{\rightarrow j} (\br)^T \bx = \br^T s_{\rightarrow (d-j)} (\bx).
\end{equation}
Thus setting $t = d - j$, we can write the above as
\[
\mathbb{E} \left[\left(\frac{1 - \sign(\br^T \mathbf{D} \bx ) \sign(\br^T \mathbf{D} \by ) }{2} - \frac{\theta}{\pi}  \right)
    \left(\frac{1 - \sign(\br^T s_{\rightarrow t}(\mathbf{D} \bx)
    ) \sign(\br^T s_{\rightarrow t}(\mathbf{D} \by) ) }{2}
    - \frac{\theta}{\pi} \right) \right] \nonumber
\]

The key idea is that we expect the vector $s_{\rightarrow t}
(\mathbf{D}\bx)$ to be nearly orthogonal to the space containing
$\mathbf{D}\bx, \mathbf{D}\by$.  This is because $\mathbf{D}$ is a
diagonal matrix of random signs, and $\bx$ and $\by$ are vectors with
small $\ell_\infty$ norm.  We show this formally in
Lemma~\ref{lem:small-dproduct}.

Why does this help?  Suppose for a moment that $\bu := s_{\rightarrow
t} (\mathbf{D}\bx)$ and $\bv := s_{\rightarrow t} (\mathbf{D}\by)$ are
both orthogonal to $\spn{\mathbf{D}\bx, \mathbf{D}\by}$.  Then for a
random Gaussian $\br$, the random variables
$\sign(\br^T \bu) \sign(\br^T \bv)$ and
$\sign(\br^T \mathbf{D}\bx) \sign(\br^T \mathbf{D}\by)$ are
independent, because the former depends only on the projection of
$\br$ onto $\spn{ \bu, \bv}$, while the latter depends only on the
projection of $\br$ onto $\spn{ \mathbf{D}\bx, \mathbf{D}\by}$. Now if
these two spaces are orthogonal, the projections of a Gaussian vector
onto these spaces are independent (this is a fundamental property of
multidimensional Gaussians).  This implies that the expectation of the
product above is equal to the product of the expectations, which is
zero (each expectation is zero).

The key lemmma (see below) now says that even if $\bu$ and $\bv$ as
defined above are {\em nearly} orthogonal to
$\spn{\mathbf{D}\bx, \mathbf{D}\by}$, we still get a good bound on the
expectation above.

\begin{lemma}\label{lem:error-bound}
Let $\ba, \bb, \bu, \bv$ be unit vectors in $\bR^d$ such that
$\angle(\ba, \bb) = \angle(\bu, \bv) = \theta$, and let $\Pi$ be the
projector onto $\spn{\ba, \bb}$.  Suppose $\max \{ \norm{ \Pi \bu
}, \norm{\Pi \bv} \} = \delta < 1$.  Then we have
\[ \mathbb{E} \left[ \left( \frac{1 - \sign(\br^T \ba ) \sign(\br^T \bb ) }{2} - \frac{\theta}{\pi}  \right)
 \left( \frac{1 - \sign(\br^T \bu ) \sign(\br^T \bv ) }{2}
 - \frac{\theta}{\pi} \right) \right] \le 2 \delta.
\]
Here, the expectation is over the choice of $\br$.
\end{lemma}
The proof of the above lemma is moved to Appendix~\ref{app:lem:error-bound}.

We use the lemma with $\ba = \bD \bx$ and $\bb = \bD \by$.  To show
Theorem~\ref{thm:variance-cbe}, we have to prove that
\begin{equation}\label{eq:to-bound-variance}
\E \left[ \max \{\norm{\Pi \bu}, \norm{\Pi \bv} \}  \right] \le 16\rho, 
\end{equation}
where $\Pi, \bu, \bv$ are defined as in the statement of
Lemma~\ref{lem:error-bound}.  The expectation now is over the choice
of $\bD$.  This leads us to our next lemma.

\begin{lemma}\label{lem:small-dproduct} 
Let $\bp, \bq \in \mathbb{R}^d$ be vectors that satisfy $\norm{\bp}_2
= 1$ and $\norm{\bq}_\infty < \rho$ for some parameter $\rho$, and
suppose $\bD := \text{diag}(\sigma_0, \sigma_1, \dots, \sigma_{d-1})$,
where $\sigma_i$ are random $\pm 1$ signs. Then for any $0<t<d$, we
have
\[ \Pr [\iprod{ \bD\bp, s_{\rightarrow t}(\bD\bq) } > \gamma] \le e^{-\gamma^2/8\rho^2}. \] 
Note that the probability is over the choice of $\bD$.
\end{lemma}
The proof of the above lemma is moved to Appendix~\ref{app:lem:small-dproduct}. 
We remark that the lemma only assumes that $\bp$ is a unit vector, it need not have a small $\ell_\infty$ norm.

We can now complete the proof of our theorem. 
As noted above, we need to show~\eqref{eq:to-bound-variance}. 
To recall, $\Pi$ is the projector onto $\spn{ \bD \bx, \bD \by}$, and we need to bound:
\begin{equation}
\E \left[ \max \{\norm{\Pi \bu}, \norm{\Pi \bv} \}  \right] \le \E [\norm{\Pi \bu}] + \E[ \norm{\Pi \bv}]. \label{eq:to-show-var2}
\end{equation}  
Let $\bx, \bz$ be an orthonormal basis for $\spn{ \bx, \by}$; then it is easy to see that for any diagonal $\bD$ with $\pm 1$ entries on the diagonal, $\bD \bx, \bD \bz$ is an orthonormal basis for $\spn{\bD \bx, \bD \by}$.  Thus
\[ \E[\norm{\Pi \bu}] \le \E [| \iprod{\bu, \bD\bx }| + |\iprod{\bu, \bD\bz}|]. \]
Now by Lemma~\ref{lem:small-dproduct},
\[ \Pr[ |\iprod{\bu, \bD\bx}| > t\rho ] \le e^{-t^2/4}. \]
Integrating over $t$, we get $\E[ |\iprod{\bu, \bD\bx}| ] \le 4\rho$. Thus we can bound the LHS of~\eqref{eq:to-show-var2} by $16\rho$, completing the proof of the theorem. \qed

\subsection{The Johnson-Lindenstrauss Type Result} 
\label{subsec:jl}

Next, we turn to the proof of Theorem~\ref{thm:jl-large-angle}, where we wish to obtain a strong tail bound.
At a high level, the argument consists of two steps:
\begin{itemize}
\item First, show that with probability $1-\epsilon$ over the choice of $\bD$, the $k$ translates of $\bx, \by$ satisfy certain orthogonality properties (this is in the same spirit as Lemma~\ref{lem:small-dproduct}). 
\item Second, conditioned on orthogonality as above, with high probability over the choice of $\br$, we have the desired guarantee.
\end{itemize}

Next will will show the two steps respectively.  
Throughout this section, we denote by $X_0, X_1, \dots, X_{k-1}$ the $k$ shifts of $\bD \bx$, i.e., $X_i = s_{\rightarrow (i)}(\bD \bx)$; define $Y_0, \dots, Y_{k-1}$ analogously as shifts of $\bD \by$.  We will also assume that $\rho < \frac{\theta^2}{16 k \log(k/\delta)}$.

The structure we require is formally the following.
\begin{definition}[$(\gamma, k)$-orthogonality]\label{def:orthogonality}
Two sequences of $k$ unit vectors $X_0, X_1, \dots, X_{k-1}$ and $Y_0, Y_1, \dots, Y_{k-1}$ are said to be $(\gamma, k)$-orthogonal if there exists a decomposition (for every $i$)
\[ X_i  = \bu_i + \be_i ~; \quad  Y_i = \bv_i + \Bf_i \]
satisfying the following properties:
\begin{enumerate}
\item $\bu_i$ and $\bv_i$ are both orthogonal to $\spn{\bu_j, \bv_j : j\ne i}$.
\item $\max_i \{ \norm{\be_i}, \norm{\Bf_i} \} < \gamma$.
\end{enumerate}
\end{definition}

The lemma of the first step, as described earlier, is the following:
\begin{lemma}\label{lem:jl-orthogonality}
Let $\bx, \by$ be unit vectors with $\norm{\bx}_\infty, \norm{\by}_\infty \leq \rho$, and $\theta = \angle(\bx, \by)$, and let $X_i, Y_i$ be rotations of $\bD \bx, \bD \by$ respectively (as defined earlier). Then w.p. $1-\delta$ over the choice of $\bD$, the vectors $(X_i, Y_i)_{i=1}^k$ are $(\gamma, k)$ orthogonal, for  $\gamma =  4 \sqrt{\rho}$.
\end{lemma}
The proof of the lemma is quite technical, and is moved to Appendix~\ref{app:lem:jl-orthogonality}. 

Now suppose we have that the shifts $X_i, Y_i$ satisfy $(\gamma, k)$-orthogonality for some $\gamma >0$.  Suppose $\bu_i, \bv_i, \be_i, \Bf_i$ are as defined earlier.  $(\gamma, k)$-orthogonality gives us that $\norm{\be_i}, \norm{\Bf_i} < \gamma$, which is $\ll 1$.  Roughly speaking, we use this to say that {\em most of the time}, $\sign( \iprod{\br, X_i} = \iprod{\br, \bu_i} )$.  Thus determining if $\sign( \iprod{\br, X_i} ) = \sign (\iprod{\br, Y_i})$ is essentially equivalent to determining if $\sign( \iprod{\br, \bu_i} ) = \sign (\iprod{\br, \bv_i})$.  But the latter quantities, by orthogonality, are indepedent! (because the signs depend only on the projection of $\br$ onto the span of $\bu_i, \bv_i$, which is independent for different $i$).\footnote{ Again, using the property of multi-variate Gaussians that the projections onto orthogonal directions are orthogonal.}  The main lemma of the second step is the following:
\begin{lemma}\label{lem:orthog-to-thm}
Let $(X_i, Y_i)_{i=1}^k$ be a set of vectors satisfying $(\gamma, k)$-orthogonality and $\angle(X_i, Y_i) = \theta$ for all $i$.  Then for any $\delta>0$ and $k > \max\{1/\gamma, \log(4/\delta) \}$, we have
\[ \Pr \left[  \left| \frac{1}{k} \sum_{i} (\sign{\iprod{\br, X_i}} \ne \sign{ \iprod{\br, Y_i} }) - \frac{\theta}{\pi} \right| > \gamma \cdot (12 \log (2k/\delta)) \right] < 1-\delta. \]
The probability here is over the choice of $\br$.
\end{lemma}
The proof is deferred to Appendix~\ref{app:lem:orthog-to-thm}. 

We can now complete the proof of Theorem~\ref{thm:jl-large-angle}.  It essentially follows  using Lemma~\ref{lem:jl-orthogonality} and Lemma~\ref{lem:orthog-to-thm}.  Note that we can apply Lemma~\ref{lem:orthog-to-thm} because the angle between $X_i$ and $Y_i$ is also $\theta$ for each $i$ (since they are shifts of $\bx, \by$).

Formally, using the value of $\gamma$ defined in Lemma~\ref{lem:jl-orthogonality}, we have that the vectors $X_i, Y_i$ are $(\gamma, k)$ orthogonal with probability $1-\delta$.  Conditioned on this, the probability that the conclusion of Lemma~\ref{lem:orthog-to-thm} holds with probability $1-\delta$.  Thus the overall probability of success is at least $1-2\delta$.
The theorem is thus easily proved by plugging in the value of $\gamma$ from Lemma~\ref{lem:jl-orthogonality}, together with $\rho < 1$.  This completes the proof of the Theorem.

\section{Optimized Binary Embedding}\label{sec:opt}
In the previous section, we showed the randomized CBE has LSH-like angle preserving properties, especially for high-dimensional data.  One problem with the randomized CBE method is that it does not utilize the underlying data distribution while generating the matrix $\mathbf{R}$. In the next section, we propose to learn $\mathbf{R}$ in a data-dependent fashion, to minimize the distortions due to circulant projection and binarization.

We propose data-dependent CBE (CBE-opt), by optimizing the projection matrix with a novel time-frequency alternating optimization. 
We consider the following objective function in learning the $d$-bit CBE. The extension of learning $k < d$ bits will be shown in Section \ref{subsec:k}.
\begin{align}
\argmin_{\mathbf{B}, \mathbf{r}}   \quad  &|| \mathbf{B} - \mathbf{X} \mathbf{R}^T ||_F^2 +  \lambda|| \mathbf{R} \mathbf{R}^T - \mathbf{I}||_F^2 \label{eq:obj}\\
\text{s.t.}  \quad &  \mathbf{R} = \circR(\mathbf{r}), \nonumber
\end{align}
where $\mathbf{X} \in  \mathbb{R}^{N \times d}$, is the data matrix containing $n$ training points: $\mathbf{X} = [\mathbf{x}_0, \cdots, \mathbf{x}_{N-1}]^T$, and
  $\mathbf{B} \in  \{-1,1\}^{N \times d}$ is the corresponding binary code matrix.\footnote{If the data is $\ell_2$ normalized, we can set $\mathbf{B} \in  \{-1/\sqrt{d},1/\sqrt{d} \}^{N \times d}$ to make $\mathbf{B}$ and $\mathbf{X}\mathbf{R}^T$ more comparable. This does not empirically influence the performance.}

In the above optimization, the first term minimizes distortion due to binarization. The second term tries to make the projections (rows of $\mathbf{R}$, and hence the corresponding bits) as uncorrelated as possible. In other words, this helps to reduce the redundancy in the learned code. 
%
If $\mathbf{R}$ were to be an orthogonal matrix, the second term will vanish and the optimization would find the best rotation such that the distortion due to binarization is minimized. 
However, being a circulant matrix, $\mathbf{R}$, in general, will not be orthogonal\footnote{We note that the rank of the circulant matrices can range from 1 (an all-1 matrix) to $d$ (an identity matrix). }. 
Similar objective has been used in previous works including \cite{gongiterative, gonglearning} and \cite{wang2010sequential}.

\subsection{The Time-Frequency Alternating Optimization}
The above is a difficult non-convex combinatorial optimization problem.
In this section we propose a novel approach to efficiently find a local solution. The idea is to alternatively optimize the objective by fixing $\mathbf{r}$, and $\mathbf{B}$, respectively. 
For a fixed $\mathbf{r}$, optimizing $\mathbf{B}$ can be easily performed in the input domain (``time'' as opposed to ``frequency''). 
For a fixed $\mathbf{B}$, the circulant structure of $\mathbf{R}$ makes it difficult to optimize the objective in the input domain. Hence we propose a novel method, by optimizing $\mathbf{r}$ in the frequency domain based on DFT. This leads to a very efficient procedure. 
%

\textbf{For a fixed $\mathbf{r}$}. 
The objective is independent on each element of $\mathbf{B}$. Denote $B_{ij}$ as the element of the $i$-th row and $j$-th column of $\mathbf{B}$. It is easy to show that $\mathbf{B}$ can be updated as:
\begin{align}
& B_{ij} = \begin{cases}
   1 & \text{if } \mathbf{R}_{j\cdot}\mathbf{x}_i \geq 0 \\
   -1 & \text{if } \mathbf{R}_{j\cdot}\mathbf{x}_i < 0
  \end{cases}, \\
&  i = 0, \cdots, N-1. \quad j = 0, \cdots, d-1. \nonumber
\end{align}
\textbf{For a fixed $\mathbf{B}$}. Define $\tilde{\mathbf{r}}$ as the DFT of the circulant vector $\tilde{\mathbf{r}} := \mathcal{F}(\mathbf{r})$. Instead of solving $\mathbf{r}$ directly, we propose to solve $\tilde{\mathbf{r}}$, from which $\mathbf{r}$ can be recovered by IDFT. 

Key to our derivation is the fact that DFT projects the signal to a set of orthogonal basis. Therefore the $\ell_2$ norm can be preserved. Formally, according to Parseval's theorem , for any $\mathbf{t} \in \mathbb{C}^d$ \cite{oppenheim1999discrete}, 
\begin{equation}
|| \mathbf{t}||_2^2 = (1/d) ||\mathcal{F}(\mathbf{t})||_2^2. 
\end{equation}

Denote $\text{diag}(\cdot)$ as the diagonal matrix formed by a vector. Denote $\Re(\cdot)$ and $\Im(\cdot)$ as the real and imaginary parts, respectively. We use $\mathbf{B}_{i\cdot}$ to denote the $i$-th row of $\mathbf{B}$. With complex arithmetic, the first term in (\ref{eq:obj}) can be expressed in the frequency domain as:

\begin{align}
& || \mathbf{B} - \mathbf{X} \mathbf{R}^T ||_F^2 = \frac{1}{d}\sum_{i = 0}^{N-1} ||\mathcal{F}(\mathbf{B}^T_{i \cdot} -  \mathbf{R} \mathbf{x}_i) ||_2^2 \label{eq:bxr}
\\ 
= & \frac{1}{d}  \sum_{i = 0}^{N-1}  ||\mathcal{F} (\mathbf{B}^T_{i \cdot}) -  \tilde{\mathbf{r}} \circ \mathcal{F}(\mathbf{x}_i) ||_2^2 = \frac{1}{d}\sum_{i = 0}^{N-1} ||\mathcal{F} (\mathbf{B}^T_{i \cdot}) - \text{diag}({\mathcal{F}(\mathbf{x}_i)} ) \tilde{\mathbf{r}} ||_2^2 \nonumber \\
= & \frac{1}{d}\sum_{i = 0}^{N-1} \left(\mathcal{F} (\mathbf{B}^T_{i \cdot}) - \text{diag}({\mathcal{F}(\mathbf{x}_i)} ) \tilde{\mathbf{r}} \right)^T 
\left(\mathcal{F} (\mathbf{B}^T_{i \cdot}) - \text{diag}({\mathcal{F}(\mathbf{x}_i)} ) \tilde{\mathbf{r}} \right) \nonumber \\
= & \frac{1}{d} \Big[ \Re(\tilde{\mathbf{r}})^T \mathbf{M} \Re(\tilde{\mathbf{r}})  +  \Im(\tilde{\mathbf{r}})^T \mathbf{M} \Im(\tilde{\mathbf{r}}) 
 +  \Re(\tilde{\mathbf{r}})^T \mathbf{h} 
 +  \Im(\tilde{\mathbf{r}})^T \mathbf{g} \Big]  +  ||\mathbf{B}||_F^2, \nonumber
\end{align}
where,
\begin{align}
&\mathbf{M} =  \text{diag}\big(
\sum_{i=0}^{N-1} \Re(\mathcal{F}(\mathbf{x}_i))\circ \Re(\mathcal{F}(\mathbf{x}_i))
 +  \Im(\mathcal{F}(\mathbf{x}_i)) \circ \Im(\mathcal{F}(\mathbf{x}_i))
  \big),  \\ 
&\mathbf{h} =  -2  \sum_{i=0}^{N-1} 
\Re(\mathcal{F}(\mathbf{x}_i))  \circ 
\Re(\mathcal{F} (\mathbf{B}^T_{i \cdot})) 
 +  
\Im(\mathcal{F}(\mathbf{x}_i)) \circ
\Im(\mathcal{F} (\mathbf{B}^T_{i \cdot})),  \\ 
&\mathbf{g}  = 2 \sum_{i=0}^{N-1}
\Im(\mathcal{F}(\mathbf{x}_i)) \circ
\Re(\mathcal{F} (\mathbf{B}^T_{i \cdot}))
- 
\Re(\mathcal{F}(\mathbf{x}_i)) \circ
\Im(\mathcal{F} (\mathbf{B}^T_{i \cdot})). 
\end{align}
The above can be derived based on the following fact. For any $\mathbf{Q} \in \mathbb{C}^{d \times d}$, $\mathbf{s}$, $\mathbf{t} \in \mathbb{C}^d$,
\begin{align}
& ||\mathbf{s} - \mathbf{Q}\mathbf{t} ||_2^2 = (\mathbf{s}-\mathbf{Q}\mathbf{t})^H  (\mathbf{s}-\mathbf{Q}\mathbf{t}) \\
= & \mathbf{s}^H \mathbf{s} - \mathbf{s}^H \mathbf{Q}\mathbf{t} - \mathbf{t}^H \mathbf{Q}^H \mathbf{s} + \mathbf{t}^H \mathbf{Q}^H A \mathbf{t} \nonumber\\
= & \Re(\mathbf{s})^T \Re(\mathbf{s}) + \Im(\mathbf{s})^T \Im(\mathbf{s}) \nonumber - 2 \Re(\mathbf{t})^T (\Re(\mathbf{Q})^T \Re(\mathbf{s})  + \Im(\mathbf{Q})^T \Im(\mathbf{s})  ) \nonumber \\ & 
+ 2 \Im(\mathbf{t})^T (\Im(\mathbf{Q})^T \Re(\mathbf{s})  - \Re(\mathbf{Q})^T \Im(\mathbf{s}) )  + \Re(\mathbf{t})^T (\Re(\mathbf{Q})^T \Re(\mathbf{Q}) + \Im(\mathbf{Q})^T \Im(\mathbf{Q})) \Re(\mathbf{t})\nonumber \\
& + \Im(\mathbf{t})^T (\Re(\mathbf{Q})^T \Re(\mathbf{Q}) + \Im(\mathbf{Q})^T \Im(\mathbf{Q})) \Im(\mathbf{t}) + 2\Re(\mathbf{t})^T (\Im(\mathbf{Q})^T \Re(\mathbf{Q}) - \Re(\mathbf{Q})^T \Im(\mathbf{Q})) \Im(\mathbf{t}). \nonumber
\end{align}

For the second term in (\ref{eq:obj}), we note that the circulant matrix can be diagonalized by DFT matrix $\mathbf{F}_d$ and its conjugate transpose $\mathbf{F}_d^H$. Formally, for $\mathbf{R} = \circR(\mathbf{r})$, $\mathbf{r} \in \mathbb{R}^d$,
\begin{equation}
\mathbf{R} = (1/d) \mathbf{F}_d^H \text{diag}(\mathcal{F}(\mathbf{r}))\mathbf{F}_d.
\end{equation}
Let $\Tr(\cdot)$ be the trace of a matrix. Therefore,
\begin{align}
&||\mathbf{R} \mathbf{R}^T - \mathbf{I}||_F^2 = ||\frac{1}{d} \mathbf{F}_d^H ( \text{diag}(\tilde{\mathbf{r}})^H\text{diag}(\tilde{\mathbf{r}}) - \mathbf{I} )
 \mathbf{F}_d||_F^2 \\
  = & \Tr  \left[   \frac{1}{d} \mathbf{F}_d^H 
  ( \text{diag}(\tilde{\mathbf{r}})^H\text{diag}(\tilde{\mathbf{r}}) - \mathbf{I} )^H
  ( \text{diag}(\tilde{\mathbf{r}})^H\text{diag}(\tilde{\mathbf{r}}) - \mathbf{I} )
  \mathbf{F}_d  \right] \nonumber \\
    = & \Tr \left[   ( \text{diag}(\tilde{\mathbf{r}})^H\text{diag}(\tilde{\mathbf{r}}) - \mathbf{I} )^H
    ( \text{diag}(\tilde{\mathbf{r}})^H\text{diag}(\tilde{\mathbf{r}}) - \mathbf{I} ) \right] \nonumber \\
 = &  || \tilde{\mathbf{r}}^H \circ \tilde{\mathbf{r}} - \mathbf{1} ||_2^2  = || \Re(\tilde{\mathbf{r}})^2 + \Im(\tilde{\mathbf{r}})^2 - \mathbf{1} ||_2^2. \nonumber
 \label{eq:second}
\end{align}

Furthermore, as $\mathbf{r}$ is real-valued, additional constraints on $\tilde{\mathbf{r}}$ are needed. For any $u \in \mathbb{C}$, denote $\overline{u}$ as its complex conjugate. We have the following result \cite{oppenheim1999discrete}: For any real-valued vector $\mathbf{t} \in \mathbb{C}^d$, $\mathcal{F}(\mathbf{t})_0 \text{ is real-valued}$, and
\begin{align}
\mathcal{F}(\mathbf{t})_{d-i} = \overline{ \mathcal{F}(\mathbf{t})_{i} }, \quad i = 1, \cdots,  \lfloor d/2 \rfloor. 
\end{align}

From (\ref{eq:bxr}) $-$ (\ref{eq:second}), the problem of optimizing $\tilde{\mathbf{r}}$ becomes
\begin{align}
\argmin_{\tilde{\mathbf{r}}} \quad & 
\Re(\tilde{\mathbf{r}})^T \mathbf{M} \Re(\tilde{\mathbf{r}}) + \Im(\tilde{\mathbf{r}})^T \mathbf{M} \Im(\tilde{\mathbf{r}}) 
+ \Re(\tilde{\mathbf{r}})^T \mathbf{h} \nonumber
\\& + \Im(\tilde{\mathbf{r}})^T \mathbf{g} 
 + \lambda d || \Re(\tilde{\mathbf{r}})^2 + \Im(\tilde{\mathbf{r}})^2 - \mathbf{1} ||_2^2 \\
\text{s.t.} \quad &  \Im(\tilde{r}_0) = 0  \nonumber\\
&             \Re(\tilde{r}_i) = \Re(\tilde{r}_{d-i}), i = 1, \cdots, \lfloor d/2 \rfloor \nonumber\\
&             \Im(\tilde{r}_i) = -\Im(\tilde{r}_{d-i}) , i = 1, \cdots, \lfloor d/2 \rfloor.\nonumber
\end{align}
The above is non-convex. Fortunately, the objective function can be decomposed, such that we can solve two variables at a time. Denote the diagonal vector of the diagonal matrix $\mathbf{M}$ as $\mathbf{m}$. The above optimization can then be decomposed to the following sets of optimizations. 
\begin{align} 
&\argmin_{\tilde{r}_0} \quad  m_0 \tilde{r}_0^2
+ h_0 \tilde{r}_0 + \lambda d \left( \tilde{r}_0^2 - 1 \right) ^2, \text{ s.t. } \tilde{r}_0 = \overline{\tilde{r}_0}. \label{eq:0}\\
&\argmin_{\tilde{r}_i} \quad (m_i + m_{d-i}) ( \Re(\tilde{r}_i)^2 + \Im(\tilde{r}_i)^2 )  \label{eq:1} + 2 \lambda d \left( \Re(\tilde{r}_i)^2 + \Im(\tilde{r}_i)^2 - 1 \right) ^2 \nonumber \\
& \quad\quad\quad\quad\quad + (h_i + h_{d-i}) \Re(\tilde{r}_i) + (g_i - g_{d-i}) \Im(\tilde{r}_i),  \quad i = 1, \cdots, \lfloor d/2 \rfloor. \nonumber
\end{align}
In (\ref{eq:0}), we need to minimize a $4^{th}$ order polynomial with one variable, with the closed form solution readily available. 
In (\ref{eq:1}), we need to minimize a $4^{th}$ order polynomial with two variables.
Though the closed form solution is hard to find (requiring solution of a cubic bivariate system), a local minima can be found by gradient descent, which in practice has constant running time for such small-scale problems.
The overall objective is guaranteed to be non-increasing in each step. In practice, we find that a good solution can be reached within just 5-10 iterations.  Therefore in practice, the proposed time-frequency alternating optimization procedure has running time $\mathcal{O}(N d \log d)$. 

\subsection{Learning with Dimensionality Reduction}
\label{subsec:k}
In the case of learning $k < d$ bits, we need to solve the following optimization problem:
\begin{align}
\hspace{-0.2cm} 
\argmin_{\mathbf{B}, \mathbf{r}} \quad  & || \mathbf{B}\mathbf{P}_k  - \mathbf{X} \mathbf{P}_k^T\mathbf{R}^T  ||_F^2 +  \lambda|| \mathbf{R}\mathbf{P}_k \mathbf{P}_k^T \mathbf{R}^T - \mathbf{I}||_F^2  \\
\text{s.t.}   \quad &  \mathbf{R} = \circR(\mathbf{r}), \nonumber
\end{align}
in which 
$\mathbf{P}_k =  
\begin{bmatrix}
  \mathbf{I}_k & \mathbf{O} \\
  \mathbf{O}   & \mathbf{O}_{d-k}
 \end{bmatrix}$,
$\mathbf{I}_k$ is a $k \times k$ identity matrix, and 
$\mathbf{O}_{d-k}$ is a  $(d-k) \times (d-k)$ all-zero matrix.

In fact, the right multiplication of $\mathbf{P}_k$  can be understood as a ``temporal cut-off'', which is equivalent to a
frequency domain convolution. This makes the optimization difficult, as the objective in frequency domain can no longer be decomposed. 
To address this issue, we propose a simple solution in which $B_{ij} = 0$, $i = 0, \cdots, N-1, j = k, \cdots, d-1$ in (\ref{eq:obj}). 
Thus, the optimization procedure remains the same, and the cost is also $\mathcal{O}(N d \log d)$. We will show in experiments that this heuristic provides good performance in practice.

\section{Discussion}\label{sec:discussion}
\subsection{Limitations of the Theory for Long Codes}
As was shown in earlier works \cite{li2011hashing, gonglearning, sanchez2011high} and as we see in our experiments (Section \ref{sec:exp}), long codes are necessary for high-dimensional data for all binary embedding methods, either randomized or optimized. 

However, when the code length is too large, our theoretical analysis is not optimal. For instance, consider our variance bound when $k > \sqrt{d}$. Here the $\rho$ term always dominates, because for any vector, we have $\rho \ge 1/\sqrt{d}$ (at least one entry of a unit vector is at least $1/\sqrt{d}$). In numeric simulations, we see that the variance drops as $1/k$ for a larger range of $k$, roughly up to $d$.  A similar behavior holds in Theorem~\ref{thm:jl-large-angle}, where the condition $\rho \le \frac{\theta^2}{16 k \log(k/\delta)}$ can hold only when $k < \mathcal{O}(\sqrt{d} / \log d)$. It is an interesting open question to analyze the variance and other concentration properties for larger $k$.

\subsection{Semi-supervised Extension}
In some applications, one can have access to a few labeled pairs of similar and dissimilar data points. Here we show how the CBE formulation can be extended to incorporate such information in learning.  This is achieved by adding an additional objective term $J(\mathbf{R})$.
\begin{align}
\argmin_{\mathbf{B}, \mathbf{r}} \quad & || \mathbf{B} - \mathbf{X} \mathbf{R}^T ||_F^2 
+ \lambda ||\mathbf{R} \mathbf{R}^T - \mathbf{I}||_F^2
+ \mu J(\mathbf{R}) \\
\text{s.t.}  \quad & \mathbf{R} = \circR(\mathbf{r}), \nonumber
\end{align}
\begin{align}
J(\mathbf{R})  =  & 
\sum_{i,j \in \mathcal{M}}  ||\mathbf{R}\mathbf{x}_i  -  \mathbf{R}\mathbf{x}_j  ||_2^2
 -  
\sum_{i,j \in \mathcal{D}} ||\mathbf{R}\mathbf{x}_i  -  \mathbf{R}\mathbf{x}_j  ||_2^2.
\end{align}

Here $\mathcal{M}$ and $\mathcal{D}$ are the set of ``similar'' and ``dissimilar'' instances, respectively. The intuition is to maximize the distances between the dissimilar pairs, and minimize the distances between the similar pairs. Such a term is commonly used in semi-supervised binary coding methods \cite{wang2010sequential}. 
We again use the time-frequency alternating optimization procedure of Section \ref{sec:opt}. For a fixed $\mathbf{r}$, the optimization procedure to update $\mathbf{B}$ is the same.  For a fixed $\mathbf{B}$, optimizing $\mathbf{r}$ is done in frequency domain by expanding $J(\mathbf{R})$ as below, with similar techniques used in Section \ref{sec:opt}. 
\begin{align}
||\mathbf{R}\mathbf{x}_i - \mathbf{R}\mathbf{x}_j  ||_2^2 = (1/d) || \text{diag}(\mathcal{F}(\mathbf{x}_i)-\mathcal{F}(\mathbf{x}_j))  \tilde{\mathbf{r}} ||_2^2. 
\end{align}
Therefore,
\begin{align}
J(\mathbf{R}) = (1/d) ( \Re(\tilde{\mathbf{r}})^T \mathbf{A} \Re(\tilde{\mathbf{r}}) + \Im(\tilde{\mathbf{r}})^T \mathbf{A} \Im(\tilde{\mathbf{r}}) ),
\end{align}
where $\mathbf{A} =  \mathbf{A}_1 + \mathbf{A}_2 - \mathbf{A}_3  - \mathbf{A}_4 $, and
\begin{equation}
\mathbf{A}_1 = \sum_{(i,j) \in \mathcal{M}} \Re(\text{diag}(\mathcal{F}(\mathbf{x}_i) - \mathcal{F}(\mathbf{x}_j)))^T \Re(\text{diag} (\mathcal{F}(\mathbf{x}_i) - \mathcal{F}(\mathbf{x}_j))),
\end{equation}
\begin{equation}
 \mathbf{A}_2 = \sum_{(i,j) \in \mathcal{M}}
\Im(\text{diag} (\mathcal{F}(\mathbf{x}_i) - \mathcal{F}(\mathbf{x}_j)))^T \Im(\text{diag} (\mathcal{F}(\mathbf{x}_i)-\mathcal{F}(\mathbf{x}_j))),
\end{equation}
\begin{equation}
\mathbf{A}_3 = \sum_{(i,j) \in \mathcal{D}}
\Re(\text{diag}(\mathcal{F}(\mathbf{x}_i) - \mathcal{F}(\mathbf{x}_j)))^T \Re(\text{diag} (\mathcal{F}(\mathbf{x}_i) - \mathcal{F}(\mathbf{x}_j))),
\end{equation}
\begin{equation}
\mathbf{A}_4 =  \sum_{(i,j) \in \mathcal{D}}
\Im(\text{diag} (\mathcal{F}(\mathbf{x}_i) - \mathcal{F}(\mathbf{x}_j)))^T \Im(\text{diag} (\mathcal{F}(\mathbf{x}_i)-\mathcal{F}(\mathbf{x}_j))).
\end{equation}

Hence, the optimization can be carried out as in Section \ref{sec:opt}, where $\mathbf{M}$ in (\ref{eq:bxr}) is simply replaced by $\mathbf{M} + \mu \mathbf{A}$. The semi-supervised extension improves over the non-semi-supervised version by 2\% in terms of averaged AUC on the ImageNet-25600 dataset.

\section{Experiments}\label{sec:exp}
To compare the performance of the circulant binary embedding techniques, we conduct experiments on three real-world high-dimensional datasets used by the current state-of-the-art method for generating long binary codes \cite{gonglearning}.  The Flickr-25600 dataset contains 100K images sampled from a noisy Internet image collection. Each image is represented by a $25,600$ dimensional vector.  The ImageNet-51200 contains 100k images sampled from 100 random classes from ImageNet \cite{deng2009imagenet}, each represented by a $51,200$ dimensional vector. The third dataset (ImageNet-25600) is another random subset of ImageNet containing 100K images in $25,600$ dimensional space. All the vectors are normalized to be of unit norm. 

We compare the performance of the randomized (CBE-rand) and learned (CBE-opt) versions of our circulant embeddings with the current state-of-the-art for high-dimensional data,  \emph{i.e.}, bilinear embeddings. We use both the randomized (bilinear-rand) and learned  (bilinear-opt) versions. Bilinear embeddings have been shown to perform similarly or better than another promising technique called Product Quantization \cite{jegou2011product}. Finally, we also compare against the binary codes produced by the baseline LSH method \cite{charikar2002similarity}, which is still applicable to 25,600 and 51,200 dimensional feature but with much longer running time and much more space. We also show an experiment with relatively low-dimensional feature (2048, with Flickr data) to compare against techniques that perform well for low-dimensional data but do not scale to high-dimensional scenario. Example techniques include ITQ \cite{gongiterative}, SH \cite{weiss2008spectral}, SKLSH \cite{raginsky2009locality}, and AQBC \cite{gong2012angular}. 
 
Following \cite{gonglearning, Norouzi11, gordo2011asymmetric}, we use 10,000 randomly sampled instances for training. We then randomly sample 500 instances, different from the training set as queries. The performance (recall@1-100) is evaluated by averaging the recalls of the query instances.
The ground-truth of each query instance is defined as its 10 nearest neighbors based on  $\ell_2$ distance. For each dataset, we conduct two sets of experiments: \textit{fixed-time} where code generation time is fixed and \textit{fixed-bits} where the number of bits is fixed across all techniques. We also show an experiment where the binary codes are used for classification.

The proposed CBE method is found robust to the choice of $\lambda$ in (\ref{eq:obj}). For example, in the retrieval experiments, the performance difference for $\lambda$ = 0.1, 1, 10, is within 0.5\%. Therefore, in all the experiments, we simply fix $\lambda$ = 1. For the bilinear method, in order to get fast computation, the feature vector is reshaped to a near-square matrix, and the dimension of the two bilinear projection matrices are also chosen to be close to square. Parameters for other techniques are tuned to give the best results on these datasets.

\begin{table}
\begin{center}
\begin{tabular}{l|l|l|l}
\hline $d$ & Full projection   & Bilinear projection & Circulant projection \\ 
\hline $2^{15}$ &  $5.44 \times 10^2$ & $2.85$   & $1.11$ \\ 
\hline $2^{17}$ &  -      & $1.91 \times 10^1$  & $4.23$ \\ 
\hline $2^{20}$ (1M) &   -     & $3.76 \times 10^2$ & $3.77 \times 10^1$ \\ 
\hline $2^{24}$ &   -     & $1.22 \times 10^4$     & $8.10 \times 10^2$ \\ 
\hline $2^{27}$ (100M)    &   -    & $2.68\times 10^5$ & $8.15 \times 10^3$ \\ 
\hline 
\end{tabular}
\end{center}
\caption{Computational time (ms) of full projection (LSH, ITQ, SH \emph{etc}.), bilinear projection (Bilinear), and circulant projection (CBE). The time is based on a single 2.9GHz CPU core. The error is within 10\%. An empty cell indicates that the memory needed for that method is larger than the machine limit of 24GB. }
\label{table:time}
\end{table}

\subsection{Computational Time} 
When generating $k$-bit code for $d$-dimensional data, the full projection, bilinear projection, and circulant projection methods have time complexity $O(kd)$, $O(\sqrt{k} d)$, and $O(d \log d)$, respectively. We compare the computational time in Table \ref{table:time} on a fixed hardware. 
Based on our implementation, the computational time of the above three methods can be roughly characterized as $d^2: d\sqrt{d}: 5d\log d $.  Note that faster implementation of FFT algorithms will lead to better computational time for CBE by further reducing the constant factor. Due to the small storage requirement $\mathcal{O}(d)$, and the wide availability of highly optimized FFT libraries, CBE is also suitable for implementation on GPU. Our preliminary tests based on GPU shows up to 20 times speedup compared with CPU. In this paper, for fair comparison, we use same CPU based implementation for all the methods.

\begin{figure*}[!ht]
\centering
\subfigure[\#bit(CBE) = 3,200]
{\includegraphics[width = 3.8cm]{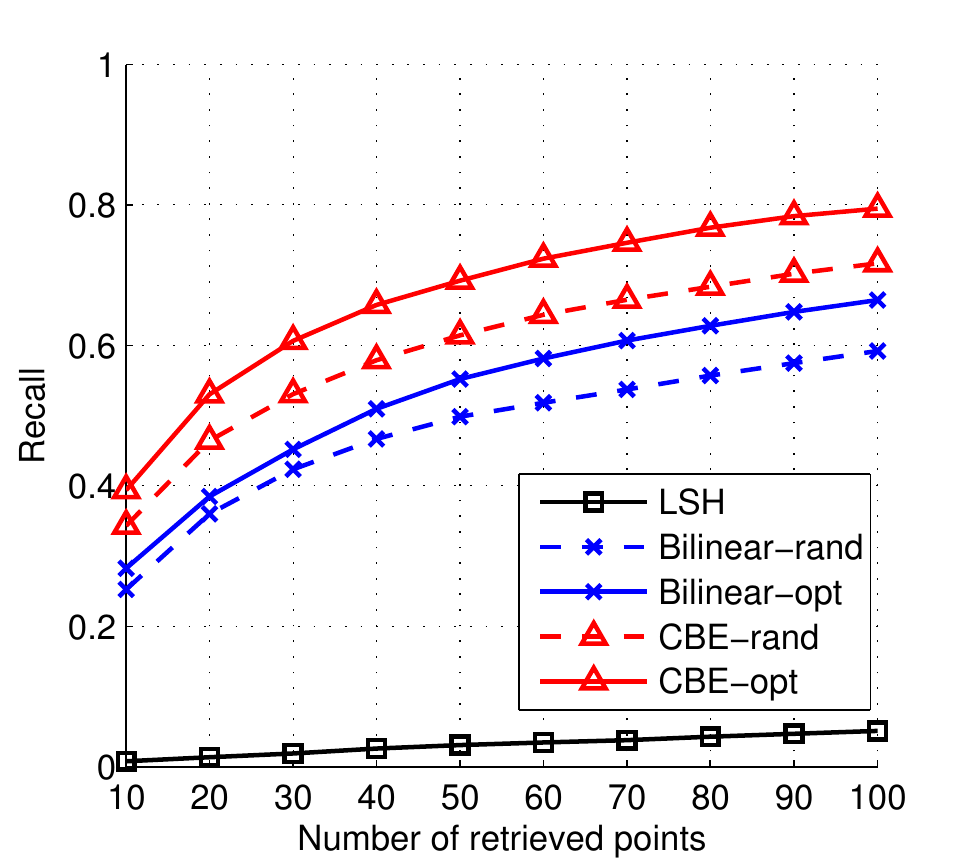}}
\hspace{-0.4cm}
\subfigure[\#bits(CBE) = 6,400]
{\includegraphics[width = 3.8cm]{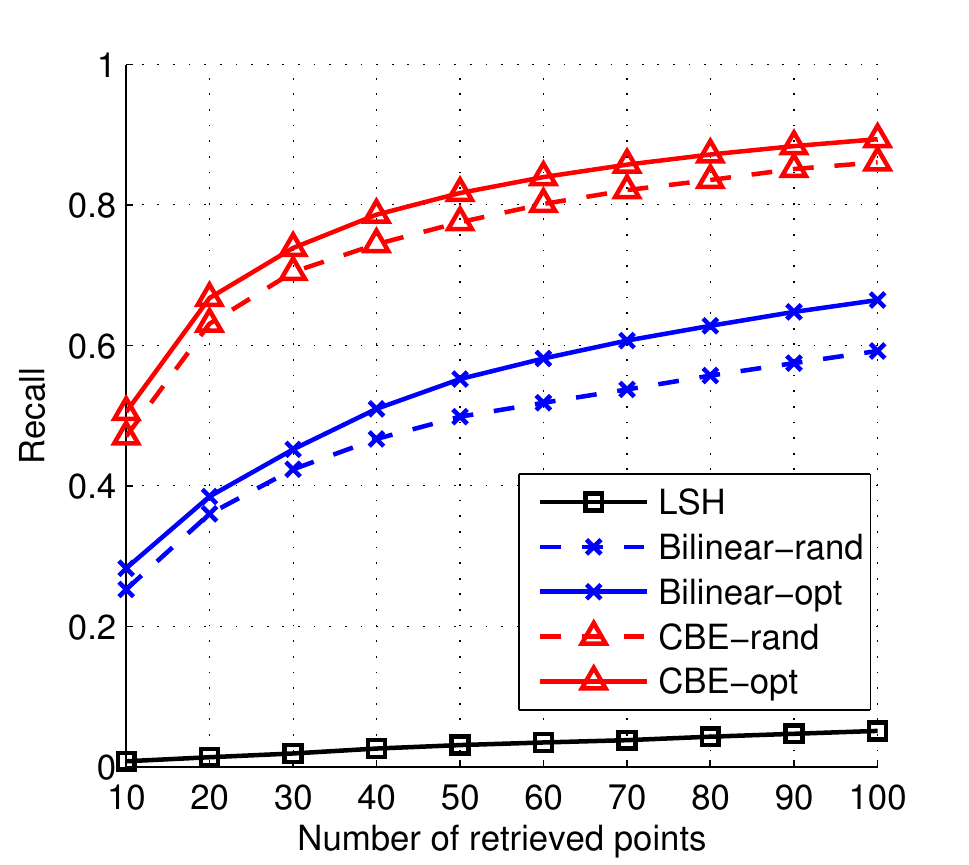}}
\hspace{-0.4cm}
\subfigure[\#bits(CBE) = 12,800]
{\includegraphics[width = 3.8cm]{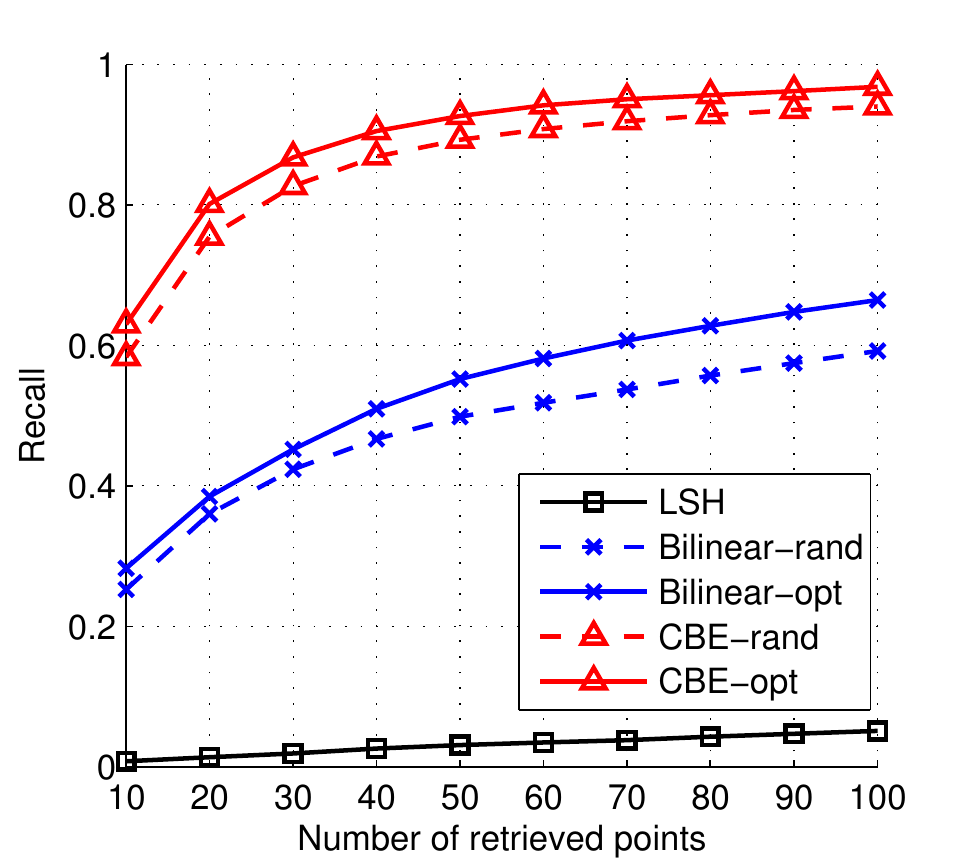}}
\hspace{-0.4cm}
\subfigure[\#bits(CBE) = 25,600]
{\includegraphics[width = 3.8cm]{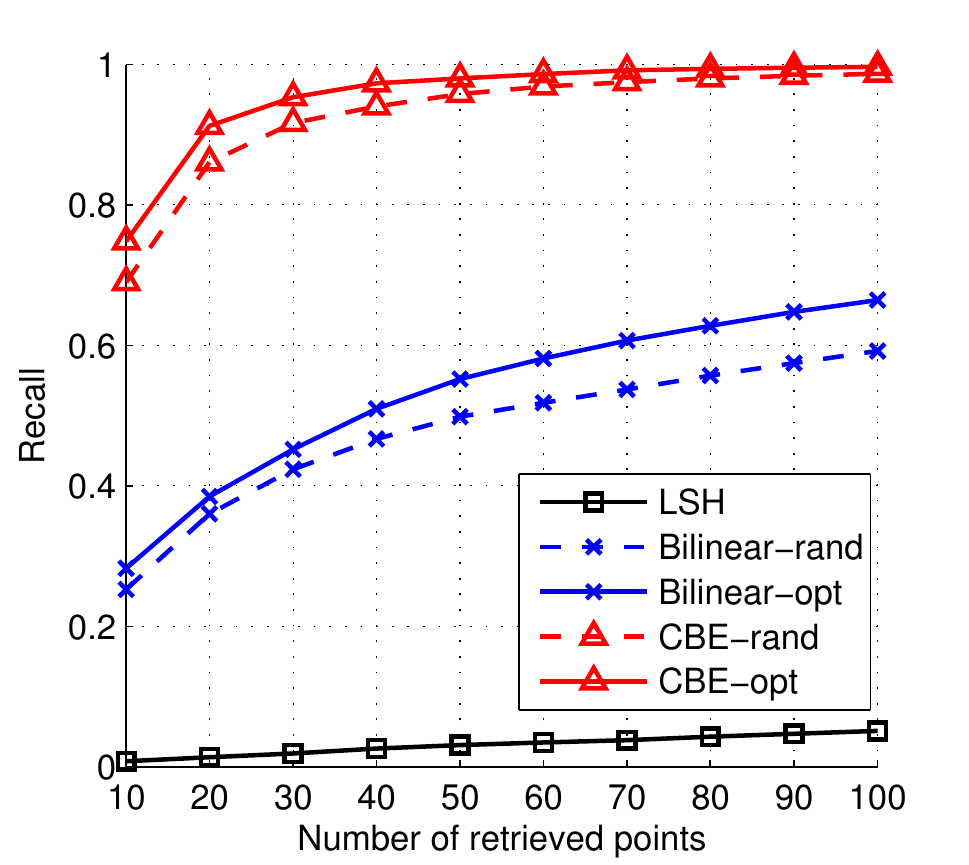}}\\
\subfigure[\# bits (all) = 3,200]
{\includegraphics[width = 3.8cm]{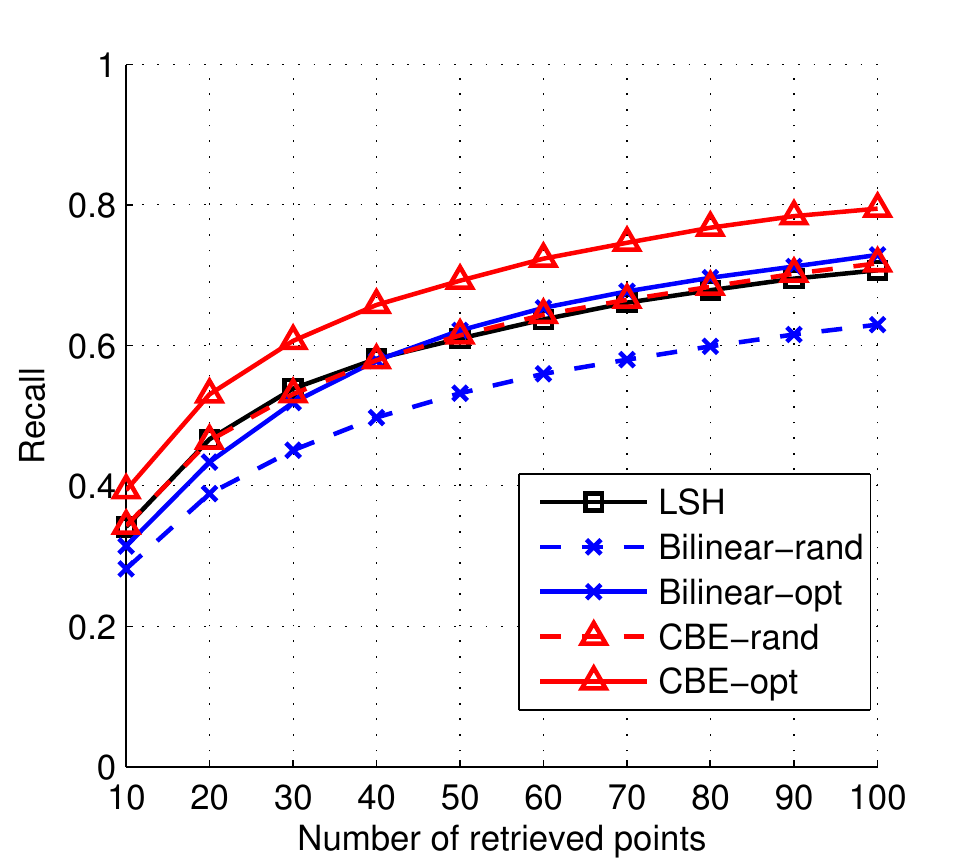}}
\hspace{-0.4cm}
\subfigure[\# bits (all) = 6,400]
{\includegraphics[width = 3.8cm]{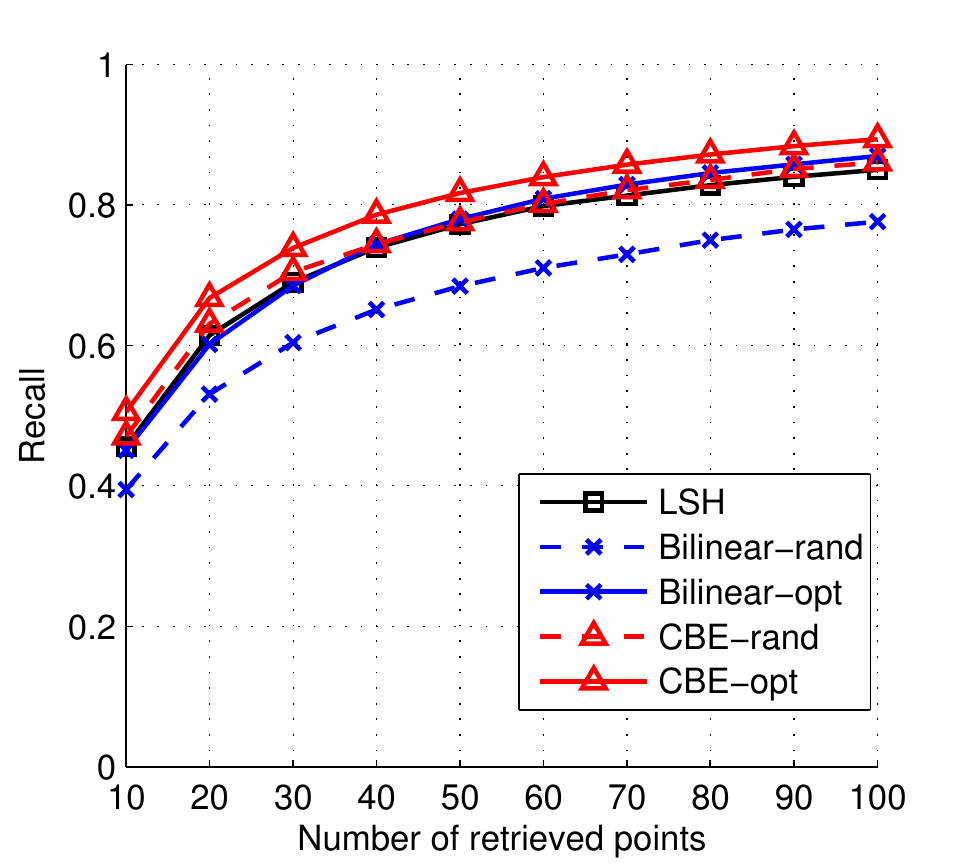}}
\hspace{-0.4cm}
\subfigure[\# bits (all) = 12,800]
{\includegraphics[width = 3.8cm]{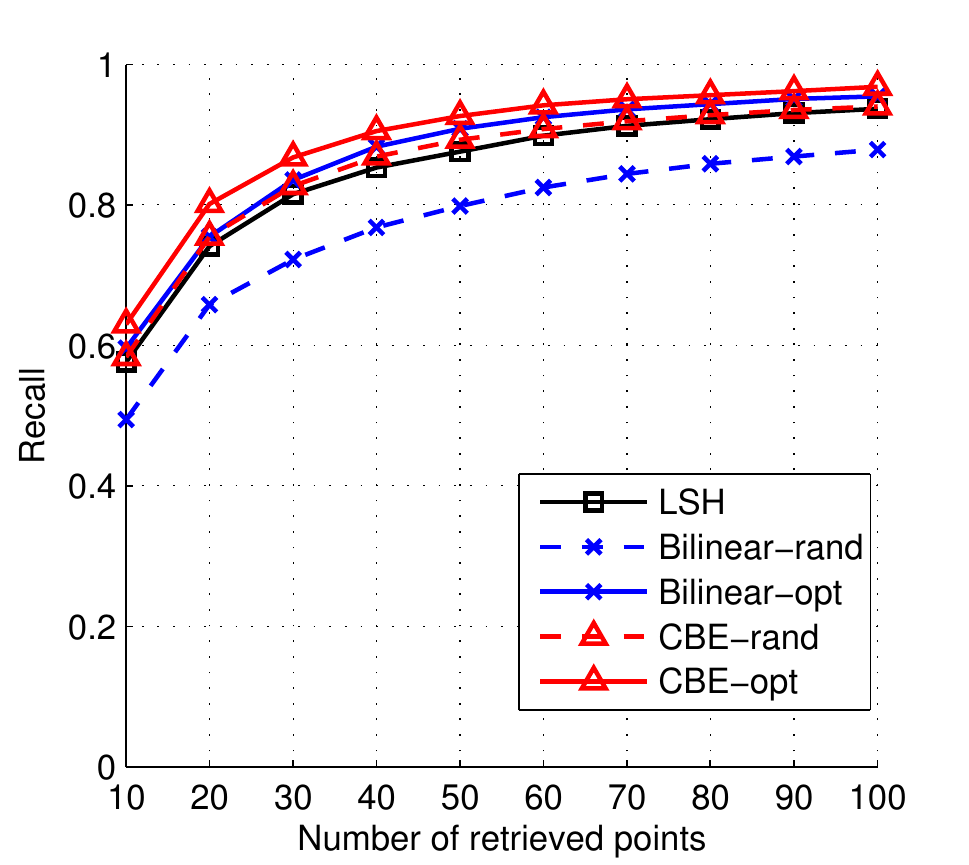}}
\hspace{-0.4cm}
\subfigure[\# bits (all) = 25,600]
{\includegraphics[width = 3.8cm]{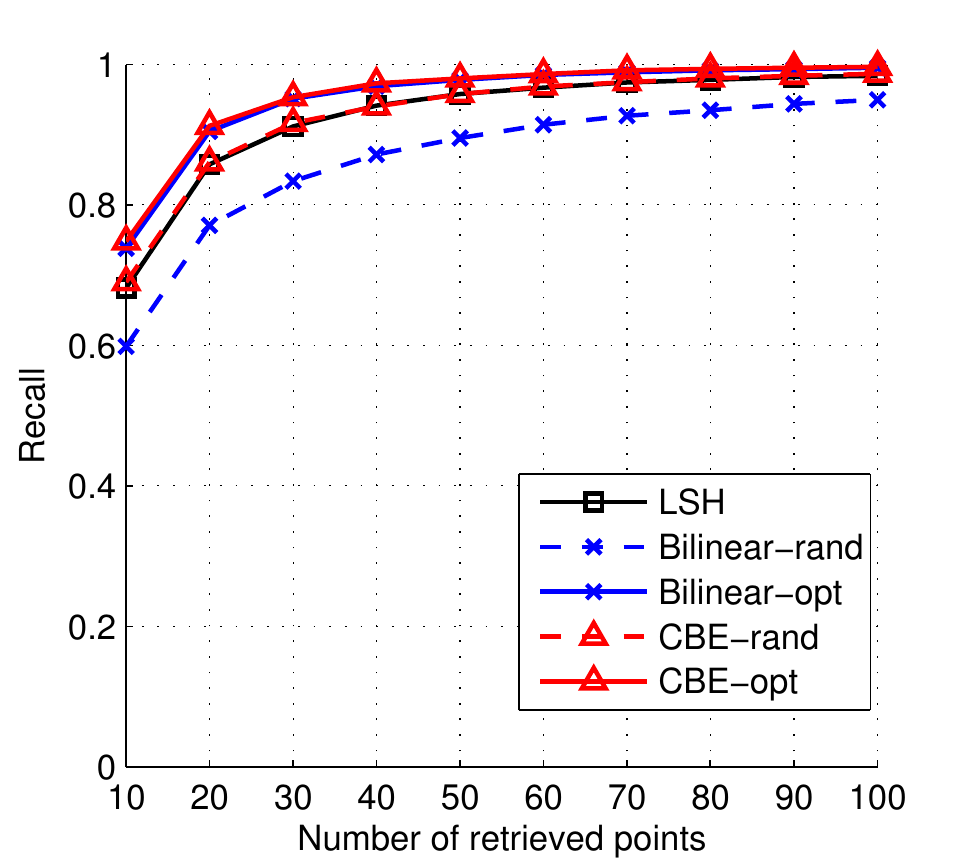}}
\caption{Recall on Flickr-25600. The standard deviation is within 1\%. \textbf{First Row}: Fixed time. ``\# bits'' is the number of bits of CBE. Other methods are using fewer bits to make their computational time identical to CBE. 
\textbf{Second Row}: Fixed number of bits. CBE-opt/CBE-rand are 2-3 times faster than Bilinear-opt/Bilinear-rand, and hundreds of times faster than LSH.}
\label{fig:flickr}
\end{figure*}

\begin{figure*}[!ht]
\centering
\subfigure[\#bits(CBE) = 3,200]
{\includegraphics[width = 3.8cm]{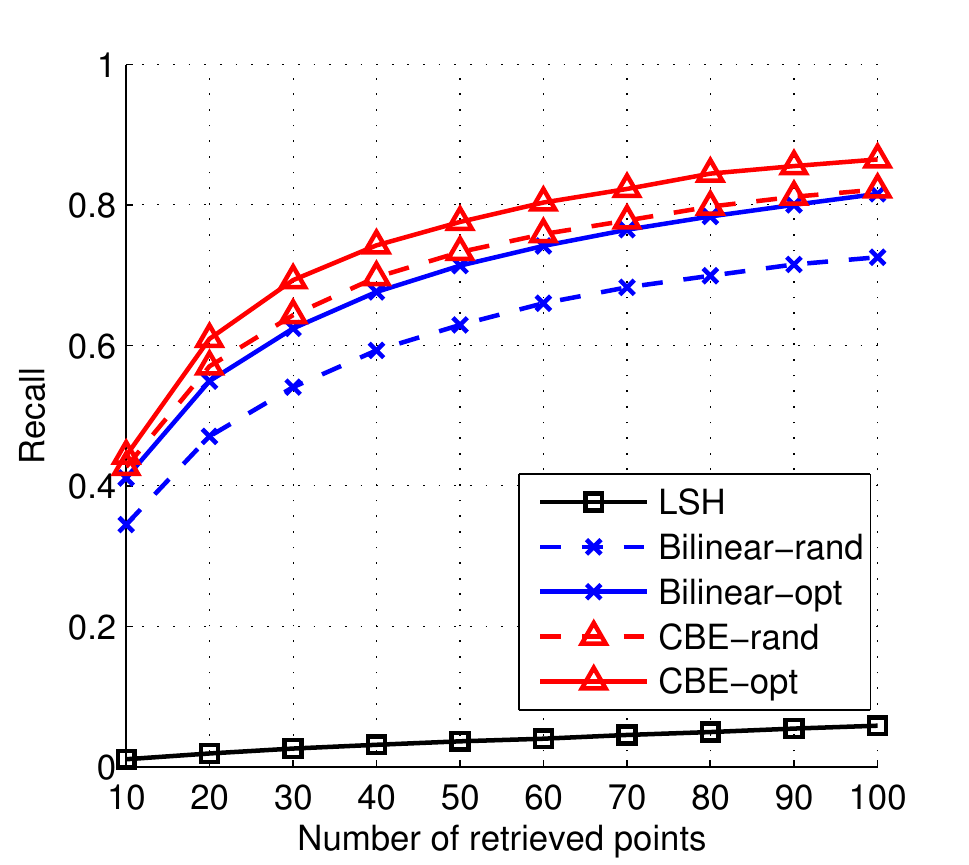}}
\hspace{-0.4cm}
\subfigure[\#bits(CBE) = 6,400]
{\includegraphics[width = 3.8cm]{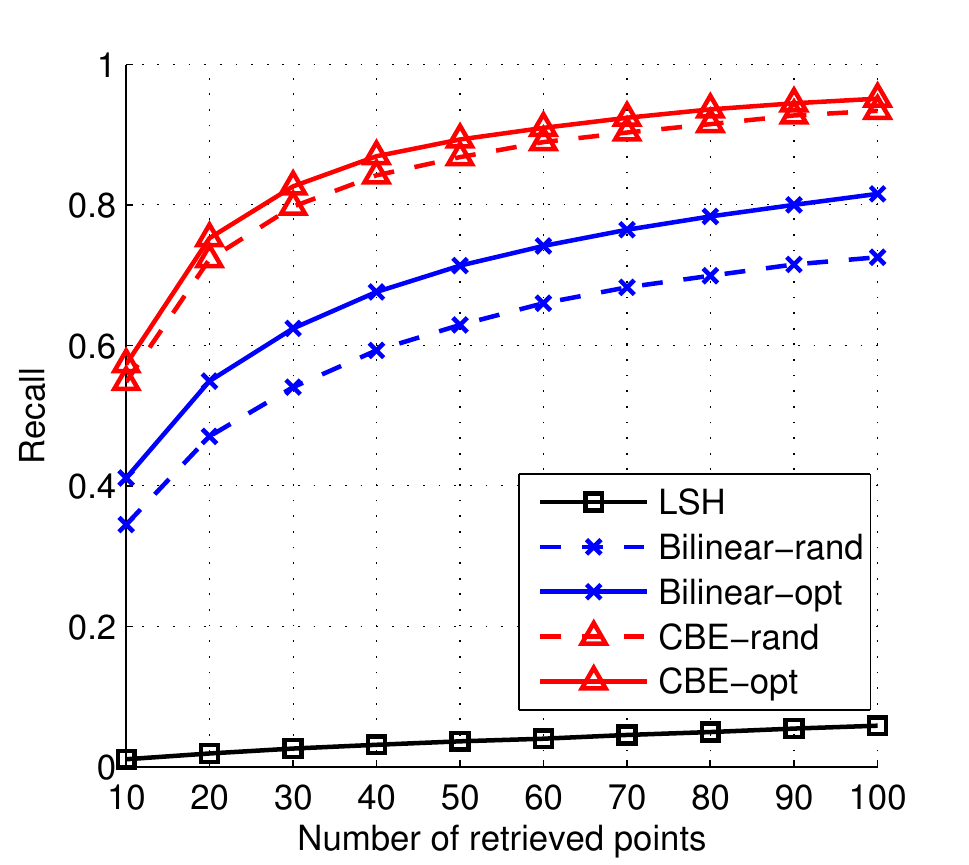}}
\hspace{-0.4cm}
\subfigure[\#bits(CBE) = 12,800]
{\includegraphics[width = 3.8cm]{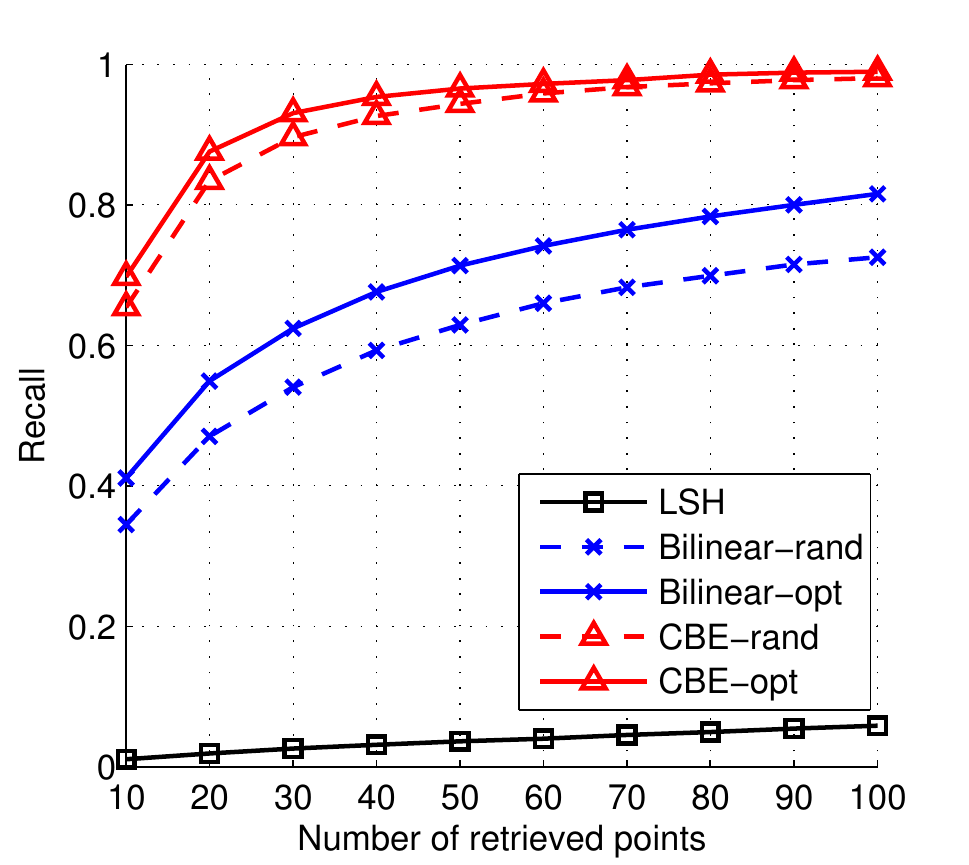}}
\hspace{-0.4cm}
\subfigure[\#bits(CBE) = 25,600]
{\includegraphics[width = 3.8cm]{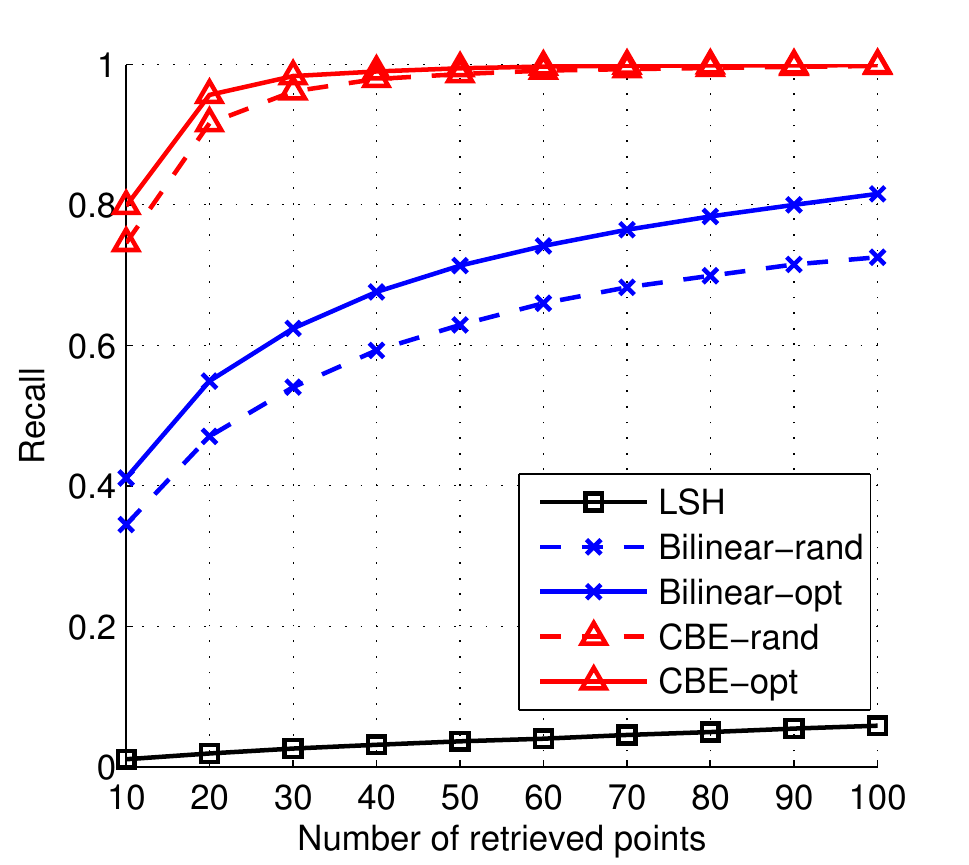}}
\\
\subfigure[\# bits (all) = 3,200]
{\includegraphics[width = 3.8cm]{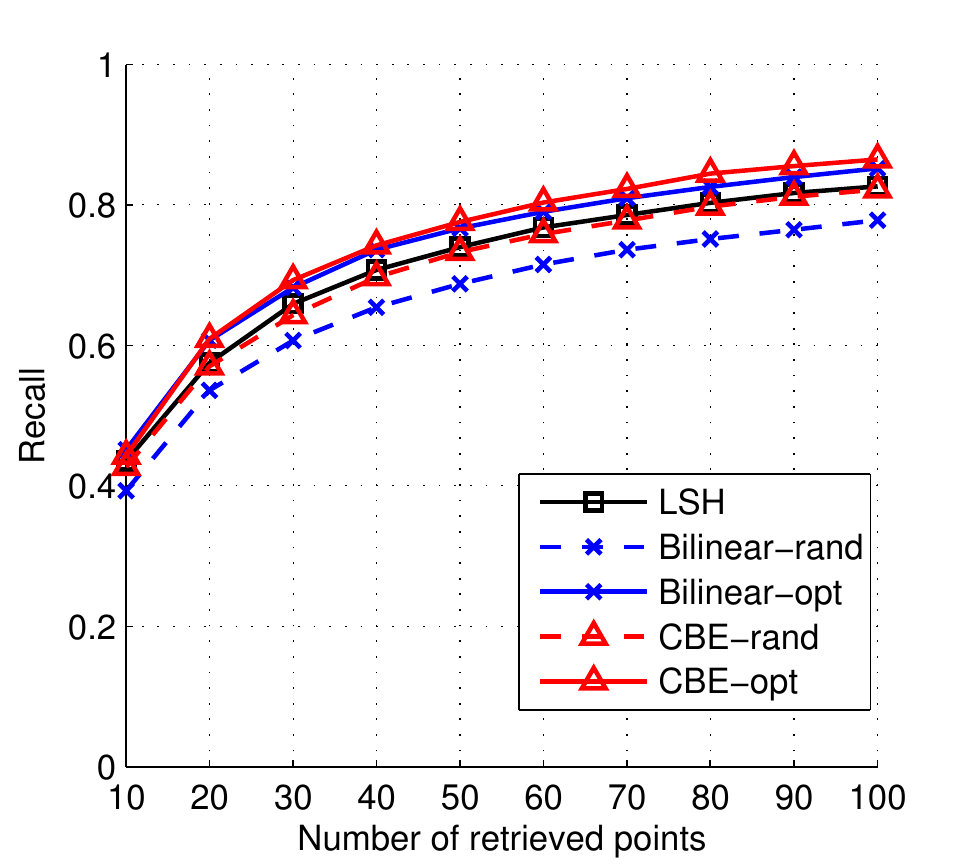}}
\hspace{-0.4cm}
\subfigure[\# bits (all) = 64,00]
{\includegraphics[width = 3.8cm]{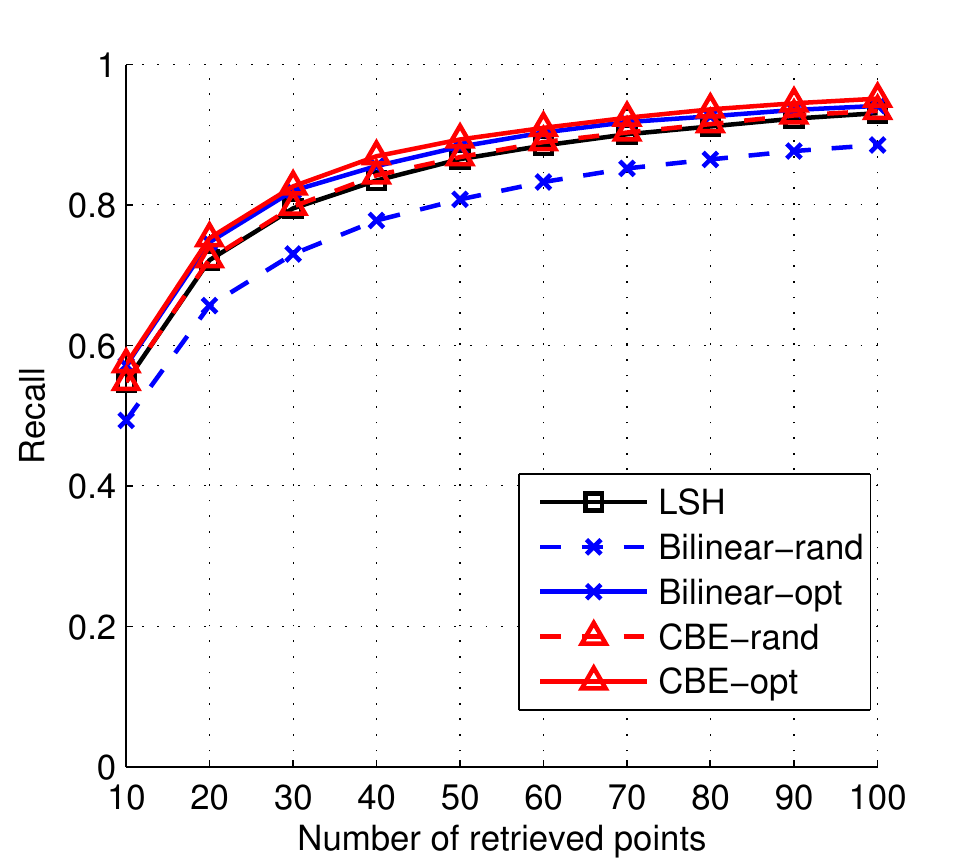}}
\hspace{-0.4cm}
\subfigure[\# bits (all) = 12,800]
{\includegraphics[width = 3.8cm]{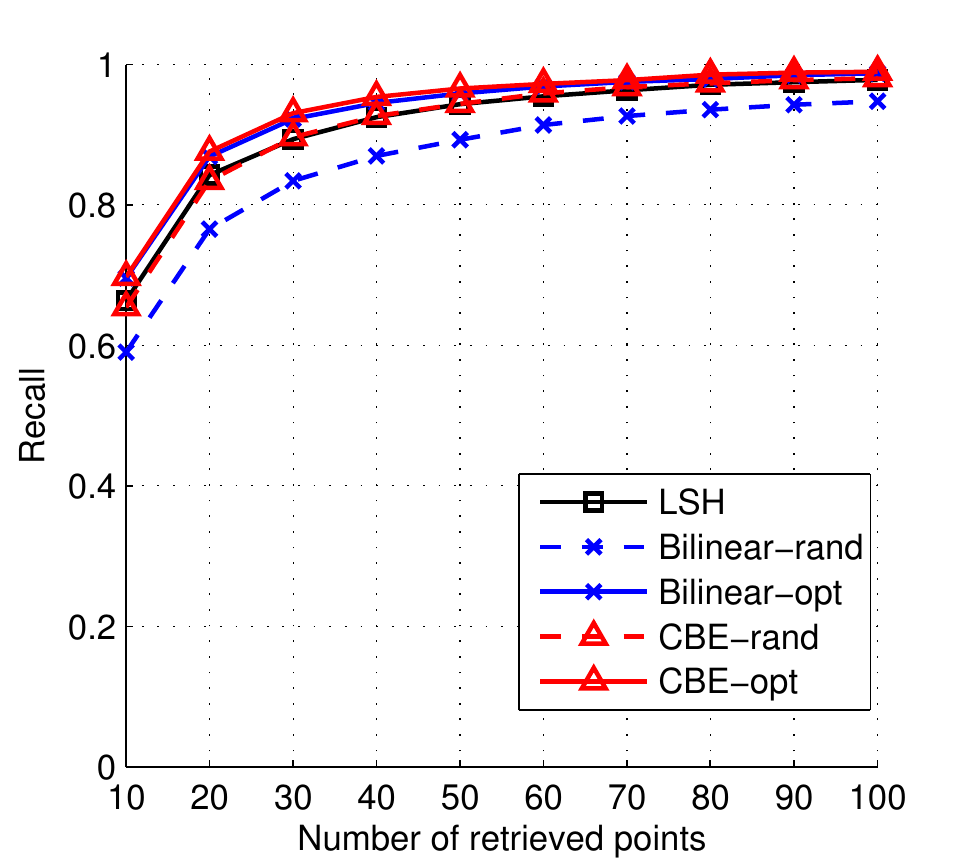}}
\hspace{-0.4cm}
\subfigure[\# bits (all) = 25,600]
{\includegraphics[width = 3.8cm]{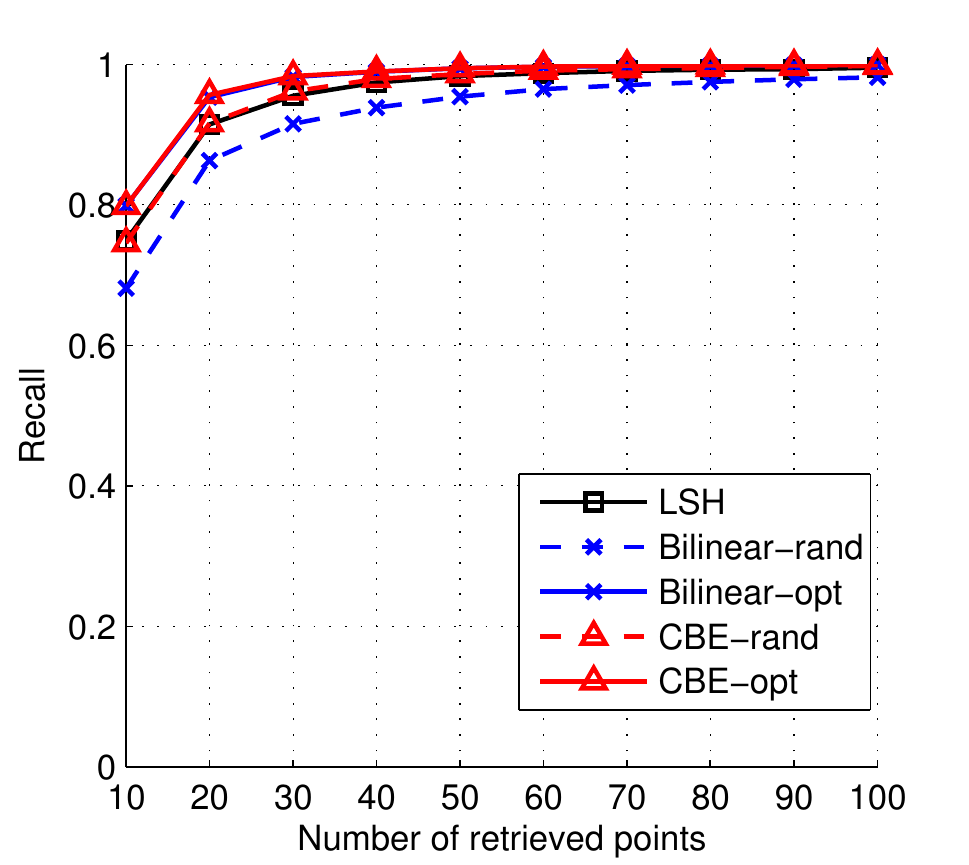}}
\caption{Recall on ImageNet-25600. The standard deviation is within 1\%. \textbf{First Row}: Fixed time. ``\# bits'' is the number of bits of CBE. Other methods are using fewer bits to make their computational time identical to CBE. 
\textbf{Second Row}: Fixed number of bits. CBE-opt/CBE-rand are 2-3 times faster than Bilinear-opt/Bilinear-rand, and hundreds of times faster than LSH.}
\label{fig:imagenet}
\end{figure*}

\begin{figure*}[!ht]
\centering
\subfigure[\#bits(CBE) = 6,400]
{\includegraphics[width = 3.8cm]{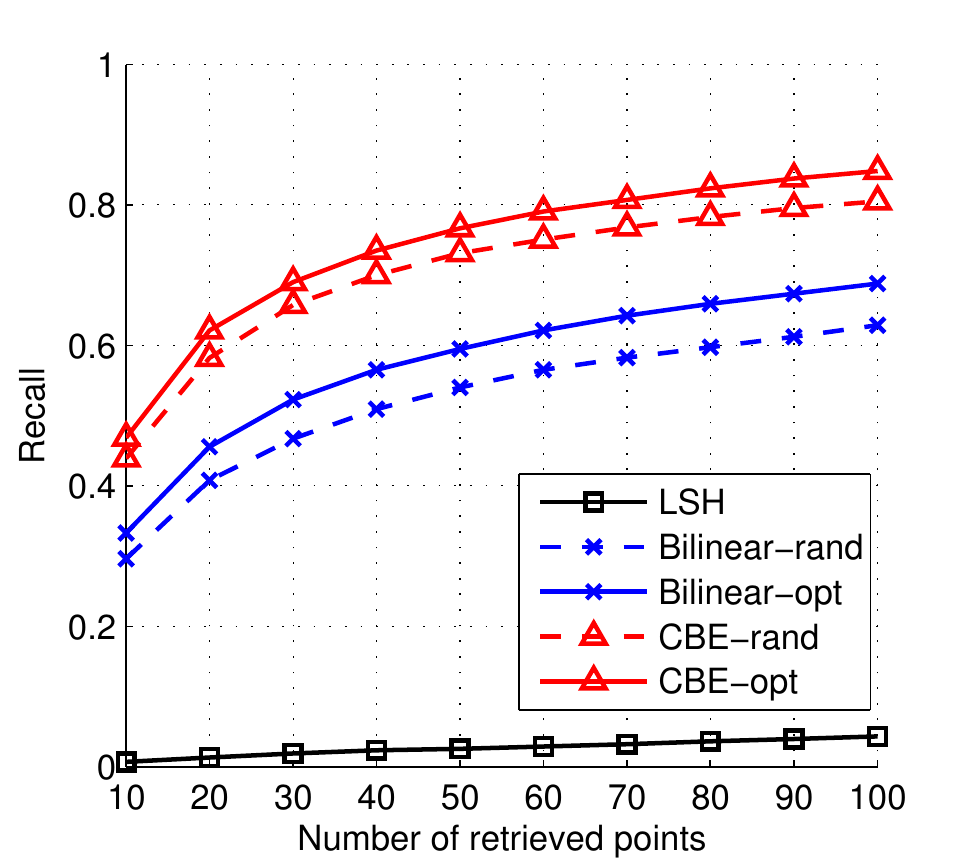}}
\hspace{-0.4cm}
\subfigure[\#bits(CBE) = 12,800]
{\includegraphics[width = 3.8cm]{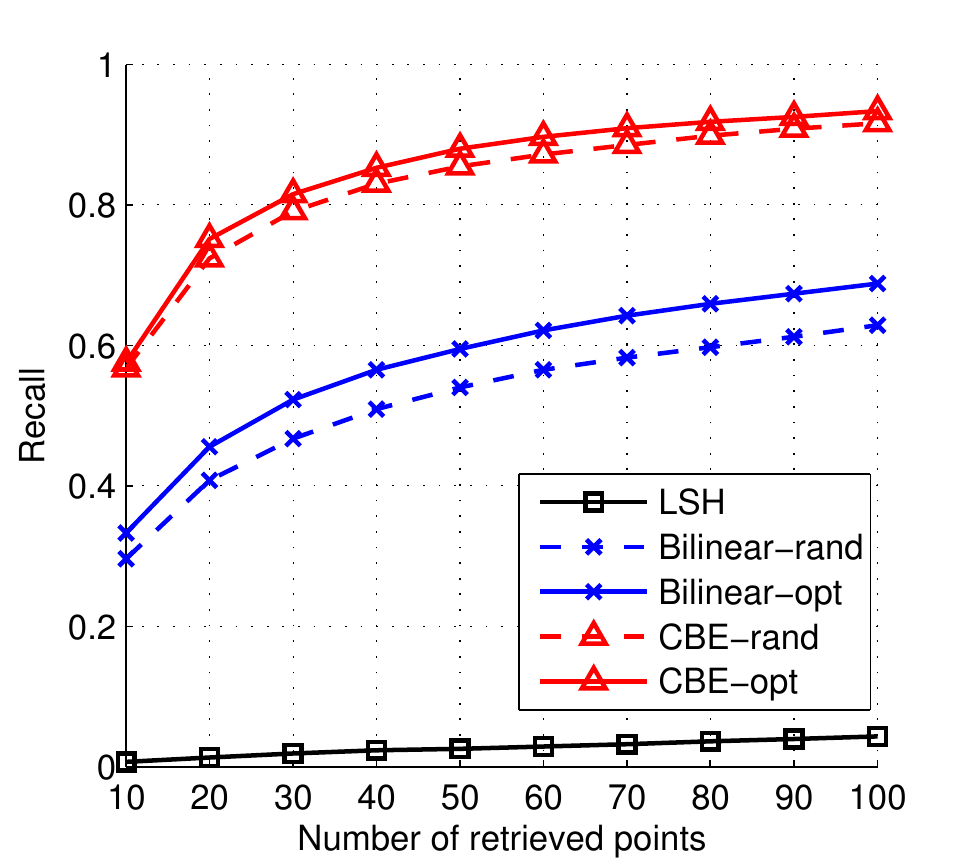}}
\hspace{-0.4cm}
\subfigure[\#bits(CBE) = 25,600]
{\includegraphics[width = 3.8cm]{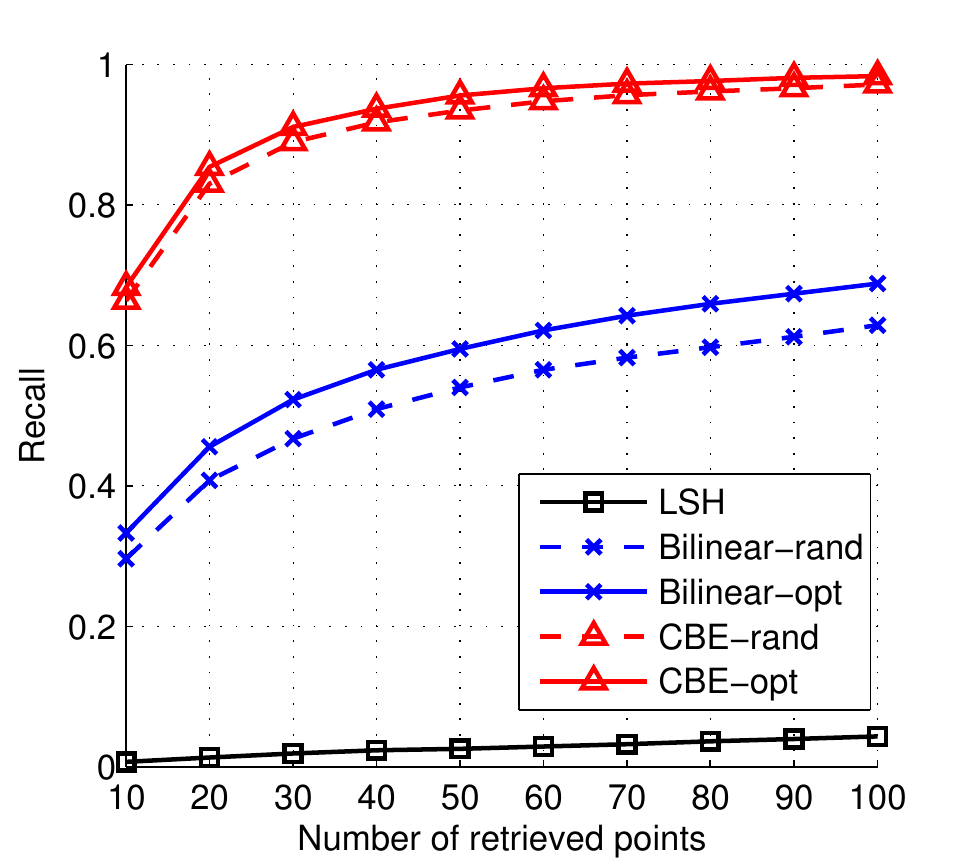}}
\hspace{-0.4cm}
\subfigure[\#bits(CBE) = 51,200]
{\includegraphics[width = 3.8cm]{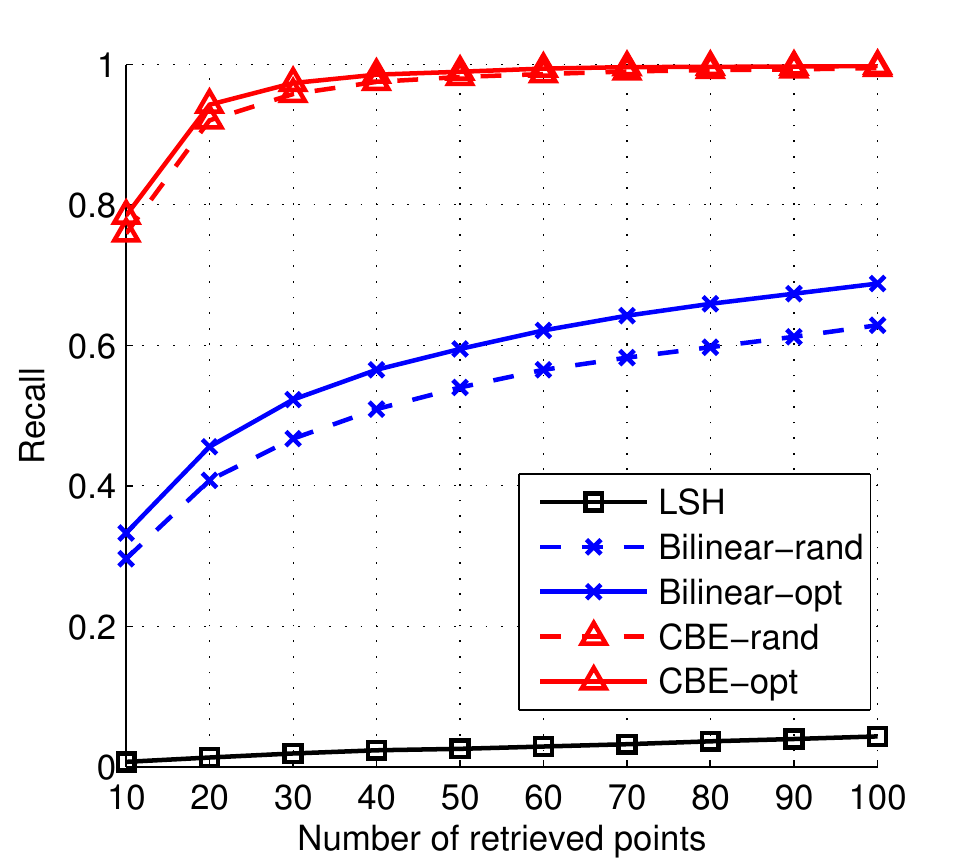}}
\\
\subfigure[\# bits (all) = 6,400]
{\includegraphics[width = 3.8cm]{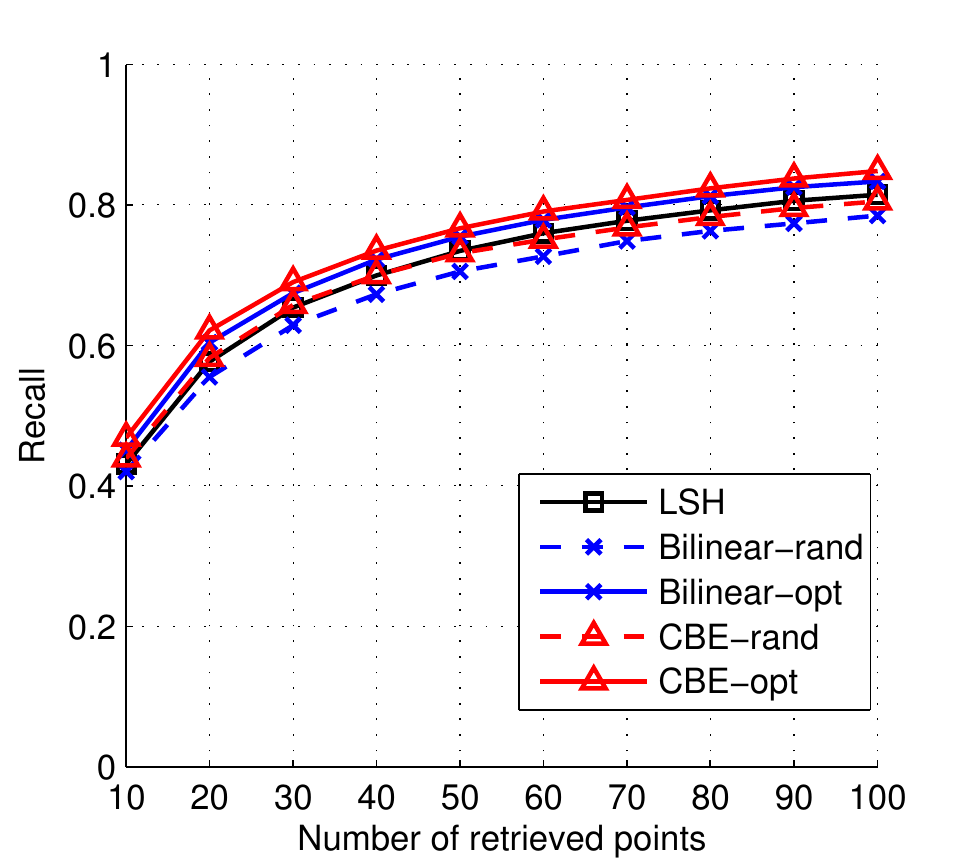}}
\hspace{-0.4cm}
\subfigure[\# bits (all) = 12,800]
{\includegraphics[width = 3.8cm]{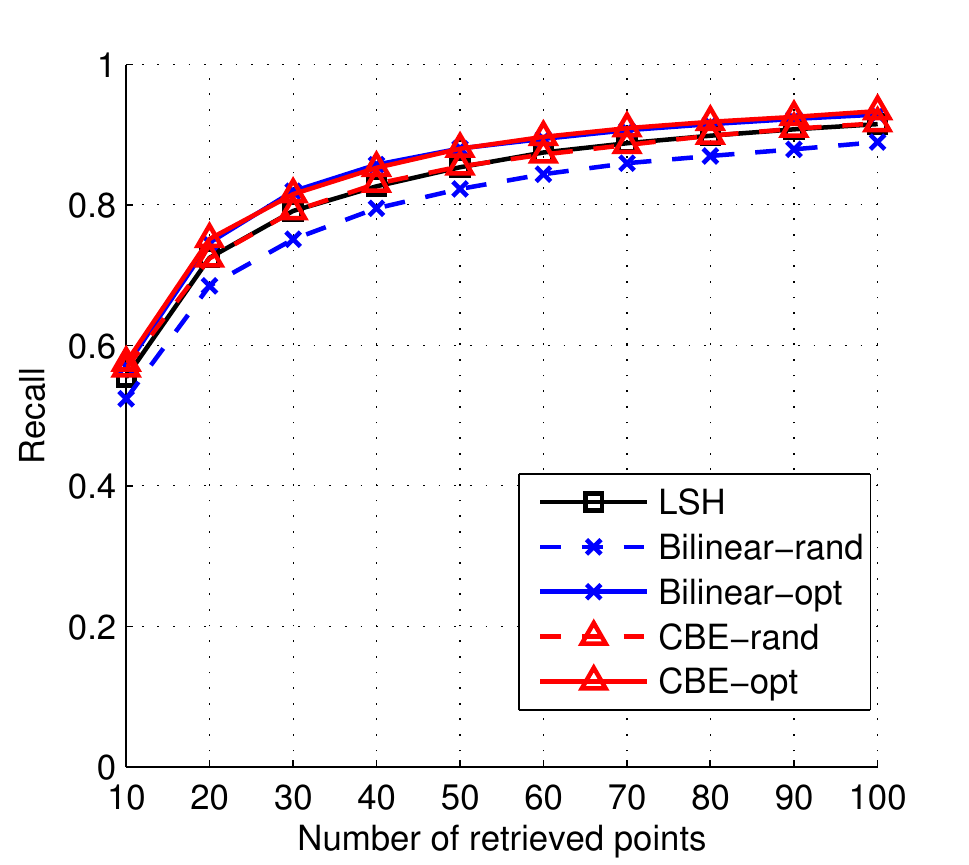}}
\hspace{-0.4cm}
\subfigure[\# bits (all) = 25,600]
{\includegraphics[width = 3.8cm]{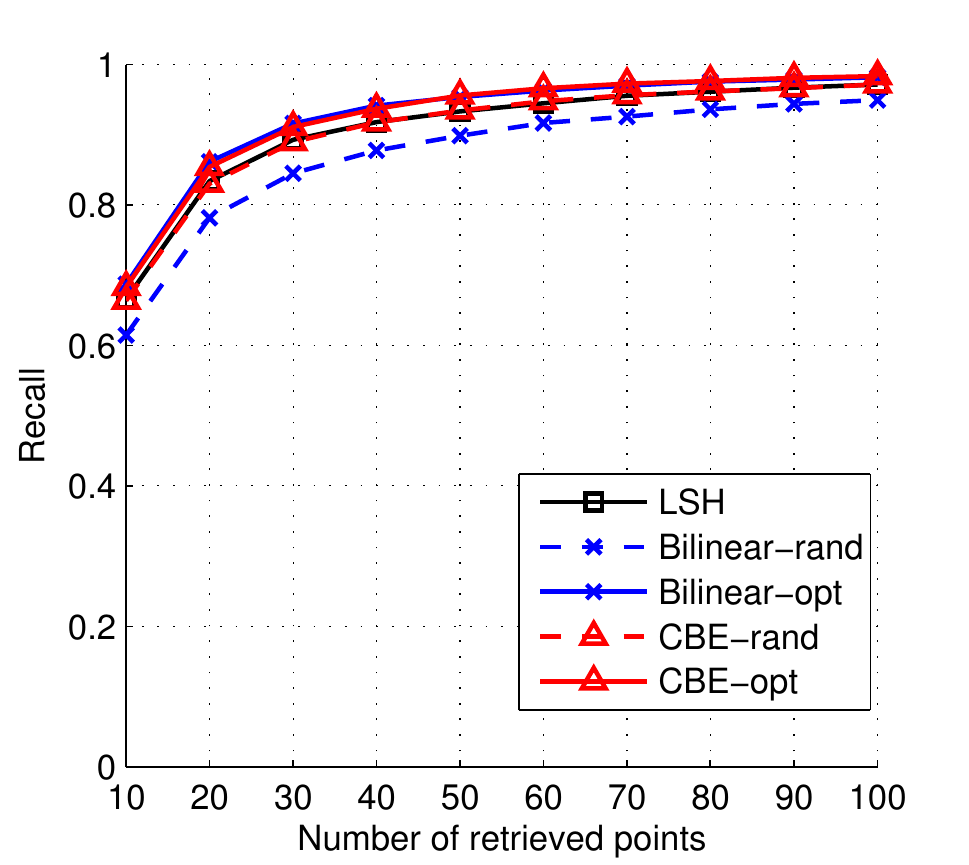}}
\hspace{-0.4cm}
\subfigure[\# bits (all) = 51,200]
{\includegraphics[width = 3.8cm]{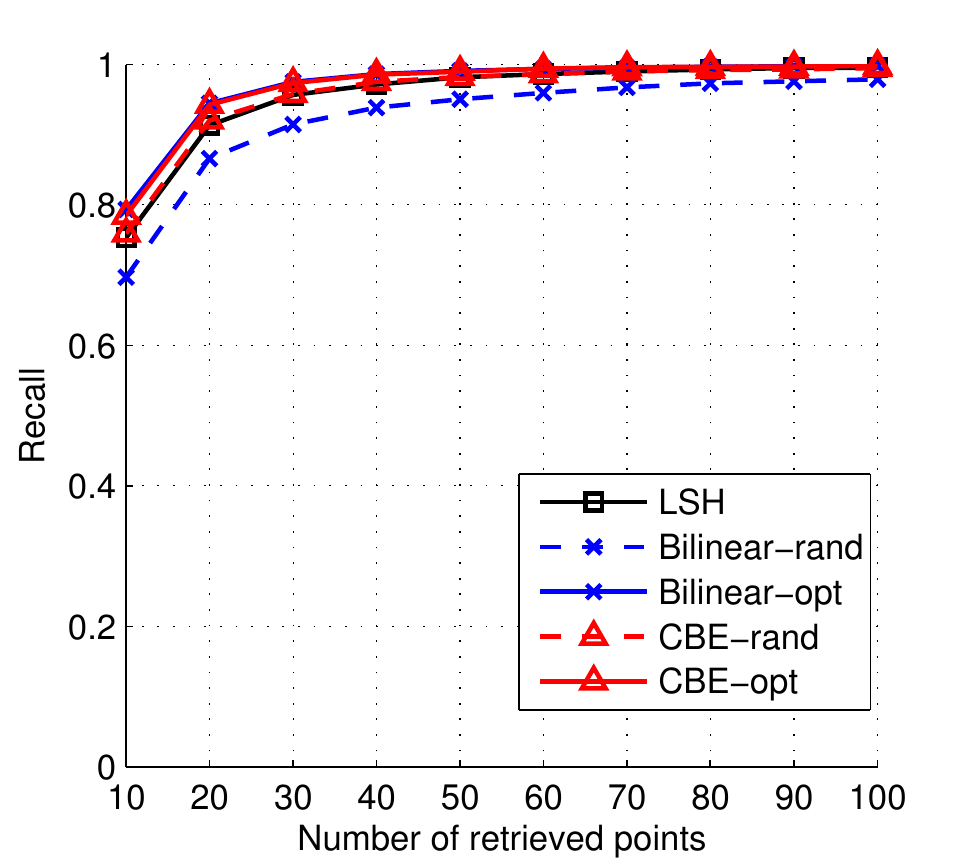}}\caption{Recall on ImageNet-51200. The standard deviation is within 1\%. \textbf{First Row}: Fixed time. ``\# bits'' is the number of bits of CBE. Other methods are using fewer bits to make their computational time identical to CBE. 
\textbf{Second Row}: Fixed number of bits. CBE-opt/CBE-rand are 2-3 times faster than Bilinear-opt/Bilinear-rand, and hundreds of times faster than LSH.}
\label{fig:imagenet_large} 
\end{figure*}

\begin{figure}[!t]
\centering
\subfigure[\# bits = 256]
{\includegraphics[width = 3.8cm]{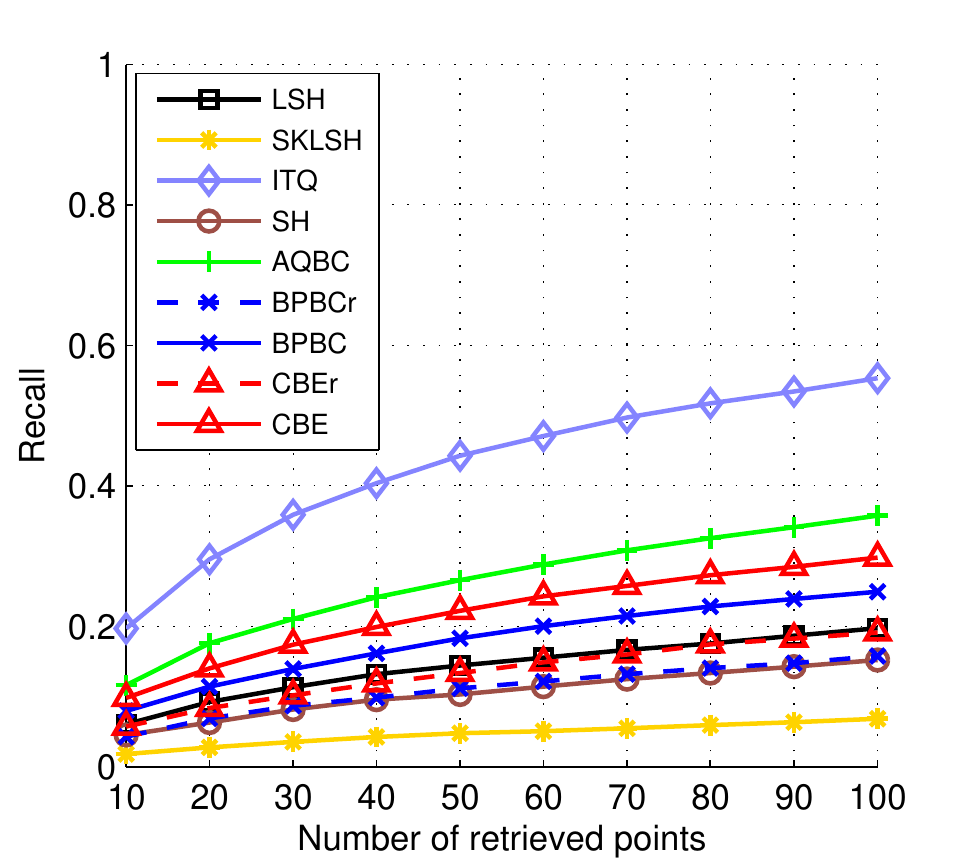}}
\hspace{-0.4cm}
\subfigure[\# bits = 512]
{\includegraphics[width = 3.8cm]{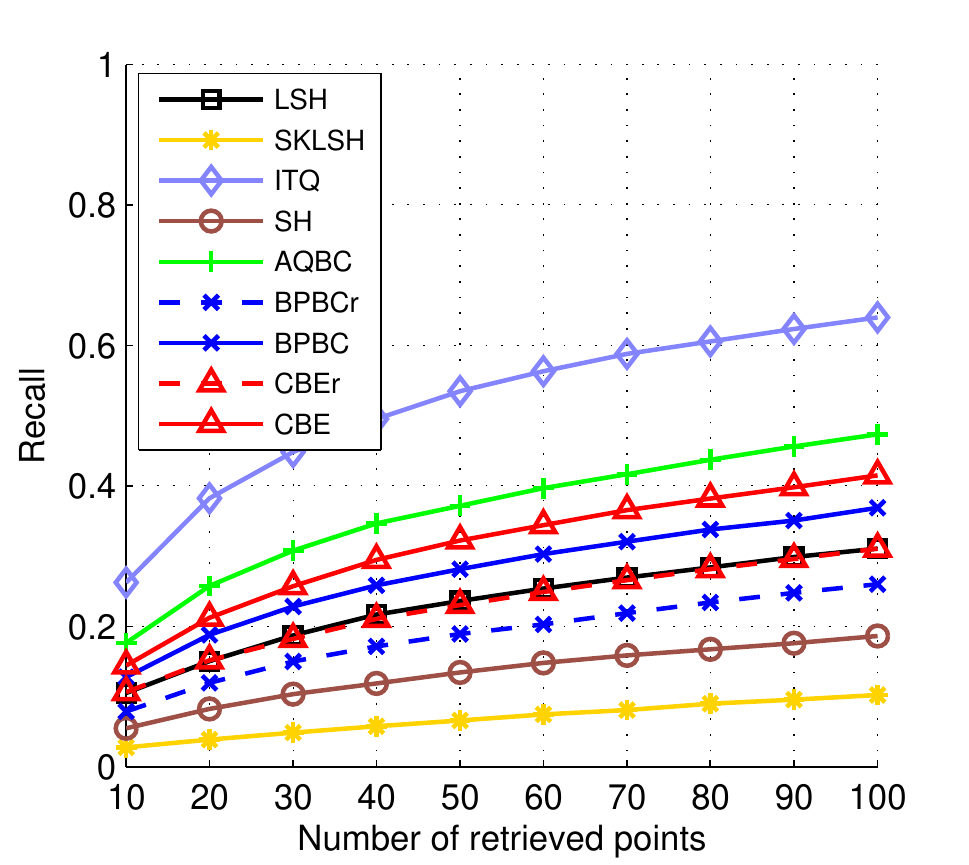}}
\hspace{-0.4cm}
\subfigure[\# bits = 1,024]
{\includegraphics[width = 3.8cm]{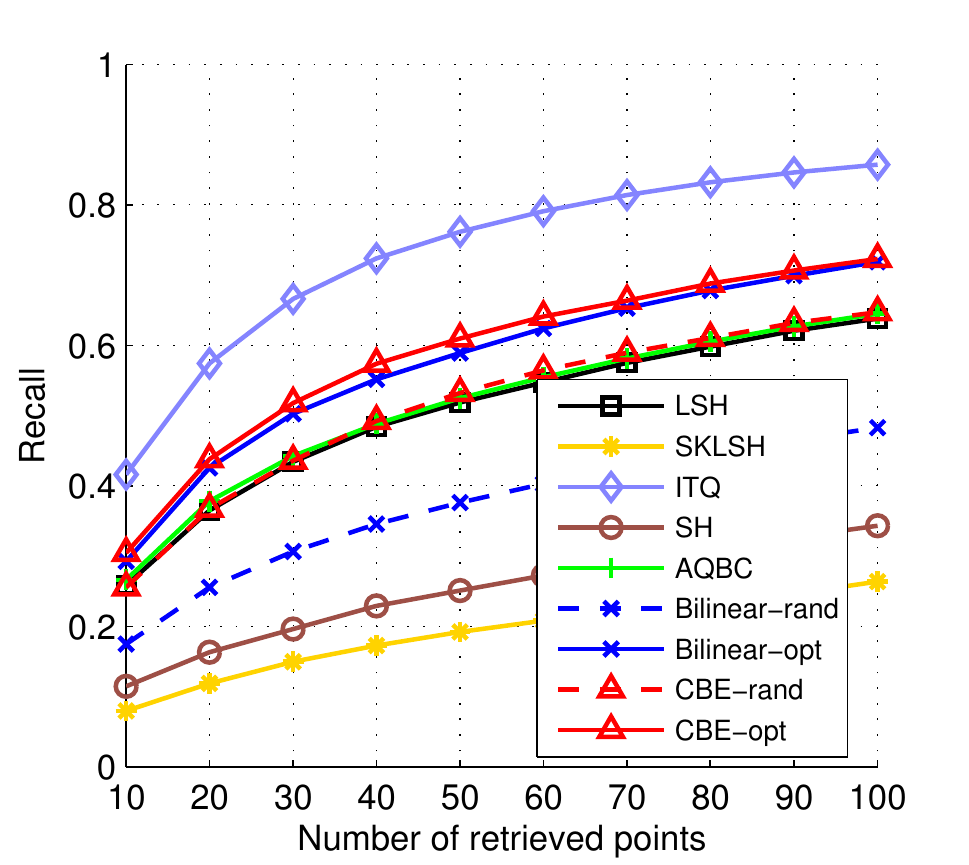}}
\hspace{-0.4cm}
\subfigure[\# bits = 2,048]
{\includegraphics[width = 3.8cm]{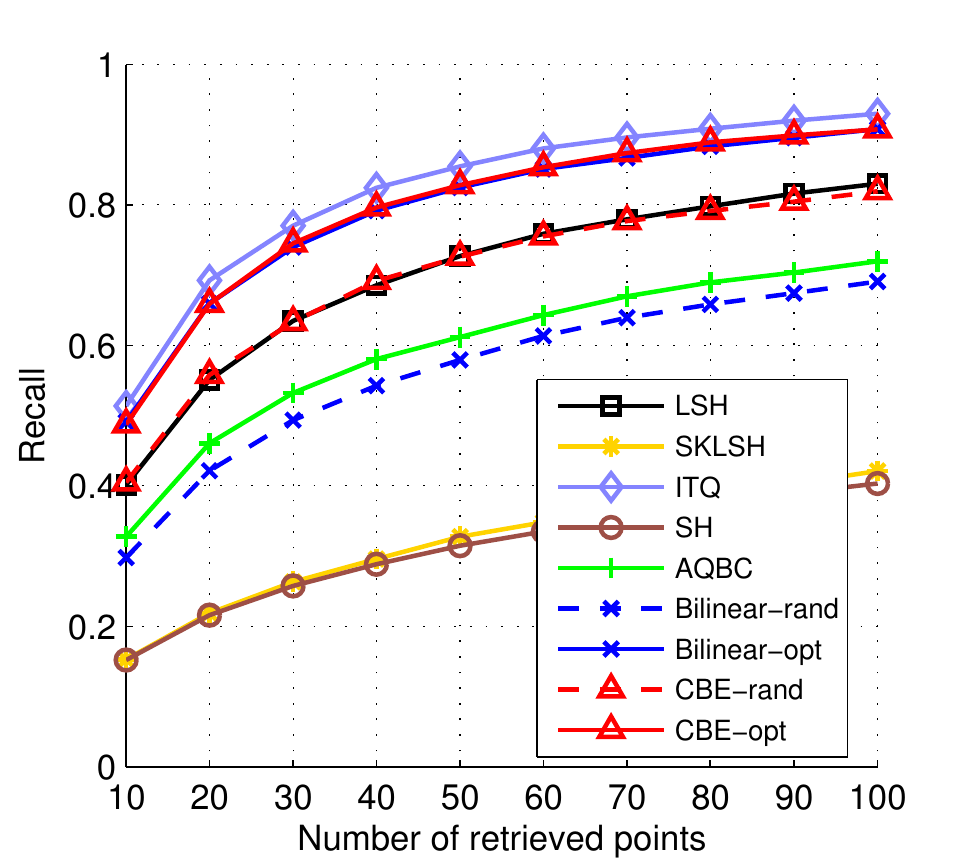}}
\caption{Performance comparison on relatively low-dimensional data (Flickr-2048) with fixed number of bits. CBE gives comparable performance to the state-of-the-art even on low-dimensional data as the number of bits is increased. However, these other methods do not scale to very high-dimensional data setting which is the main focus of this work.}
\label{fig:flickr_2048}
\end{figure}

\subsection{Retrieval} 
The recalls of different methods are compared on the three datasets, shown in 
Figure \ref{fig:flickr} -- \ref{fig:imagenet_large}. The top row in each figure shows the performance of different methods when the code generation time for all the methods is kept the same as that of CBE. For a fixed time, the proposed CBE yields much better recall than other methods. Even CBE-rand outperforms LSH and Bilinear code by a large margin. The second row compares the performance for different techniques with codes of same length. In this case, the performance of CBE-rand is almost identical to LSH even though it is hundreds of time faster. This is consistent with our analysis in Section \ref{sec:rand}. Moreover, CBE-opt/CBE-rand outperform Bilinear-opt/Bilinear-rand in addition to being 2-3 times faster. 

There exist several techniques that do not scale to high-dimensional case. To compare our method with those, we conduct experiments with fixed number of bits on a relatively low-dimensional dataset (Flickr-2048), constructed by randomly sampling 2,048 dimensions of Flickr-25600. As shown in Figure \ref{fig:flickr_2048}, though CBE is not designed for such scenario, the CBE-opt performs better or equivalent to other techniques except ITQ which scales very poorly with $d$ ($\mathcal{O}(d^3)$). Moreover, as the number of bits increases, the gap between ITQ and CBE becomes much smaller suggesting that the performance of ITQ is not expected to be better than CBE even if one could run ITQ on high-dimensional data.

\begin{figure}
\centering
\subfigure[\# bits = 5,000]
{\includegraphics[width = 5cm]{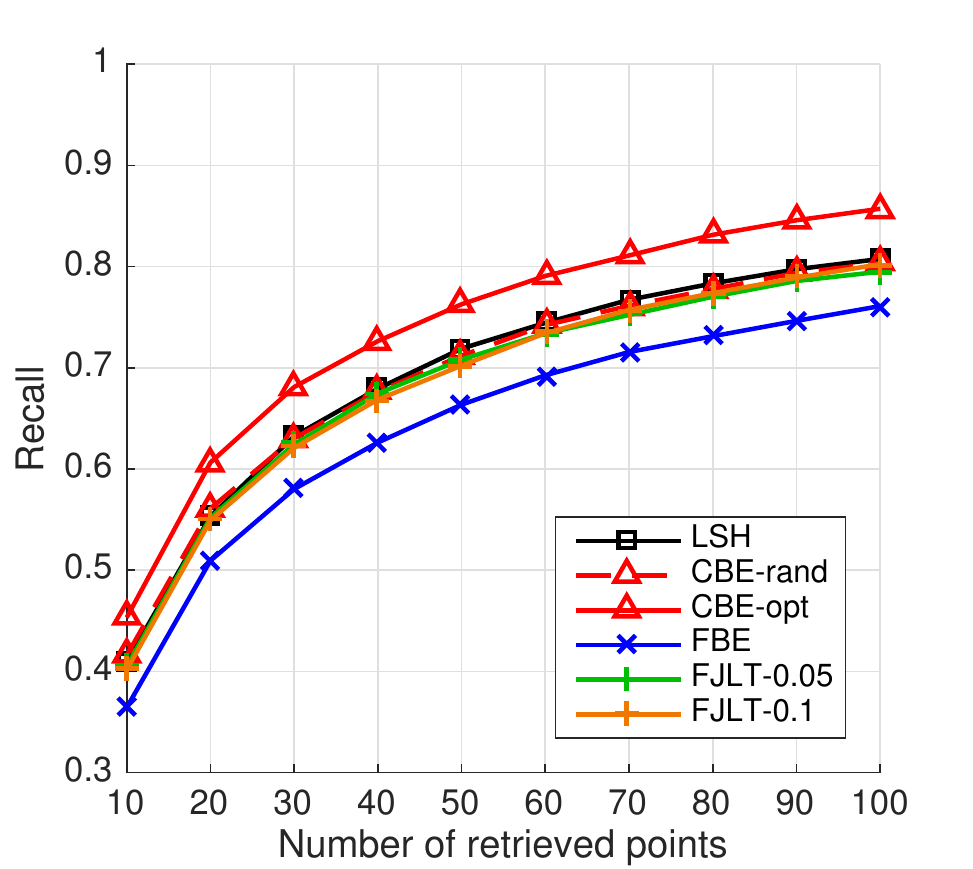}}
\hspace{-0.4cm}
\subfigure[\# bits = 10,000]
{\includegraphics[width = 5cm]{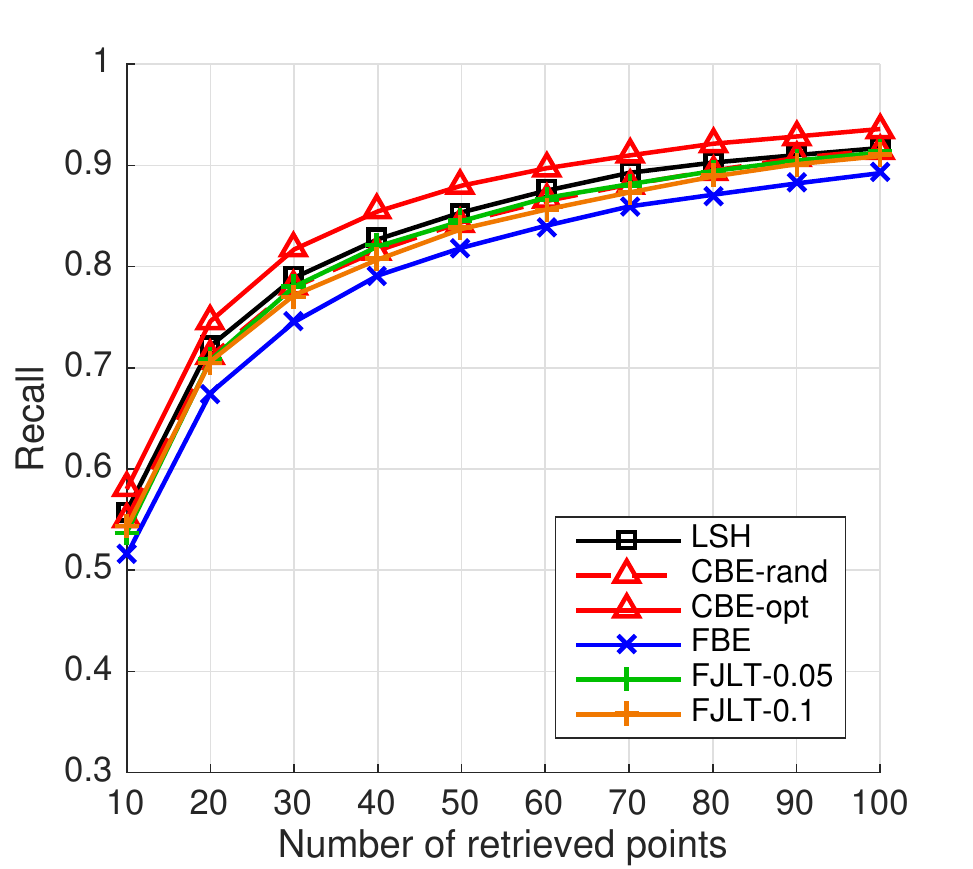}}
\hspace{-0.4cm}
\subfigure[\# bits = 15,000]
{\includegraphics[width = 5cm]{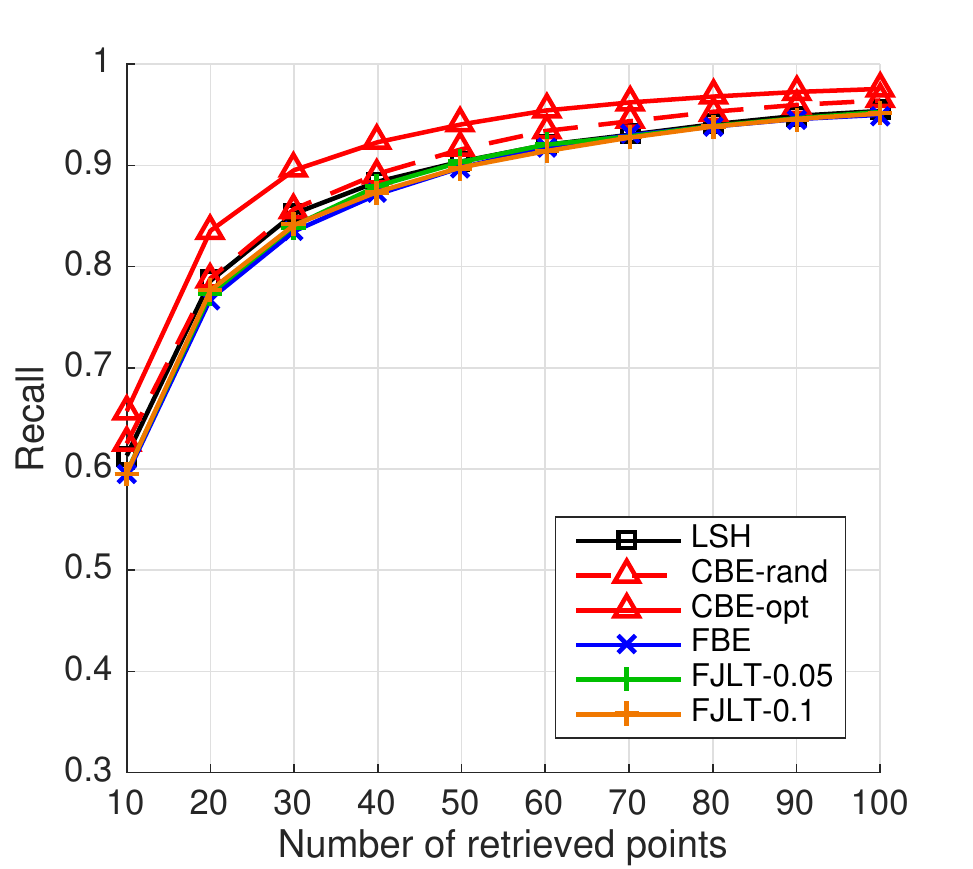}}
\caption{Recall on the Flickr-25600 dataset. The methods compared are CBE-rand, CBE-opt, Fast Binary Embedding (FBE), Fast Johnson-Lindenstruss Transformation (FJLT) based methods and LSH. We follow the detailed setting of the FBE paper\cite{yi2015binary}. In FJLT-$p$, $p$ represents the percentage of nonzero elements in the sparse Gaussian matrix.}
\label{fig:cbe_vs_fbe}
\end{figure}

We also conduct additional experiments to compare CBE with the more recent Hadamard-based algorithms. The first algorithm we consider generates the binary code using the Fast Johnson-Lindenstrss Transformation (FJLT). Similar to the circulant projection, FJLT has been used in dimensionality reduction \cite{ailon2006approximate}, deep neural networks \cite{yang2014deep}, and kernel approximation \cite{le2013fastfood}. Here, the binary code of $\bx \in \mathbb{R}^d$ is generated as
\begin{equation}
h(\bx) = \sign(\mathbf{P} \mathbf{H} \mathbf{D}\bx),
\end{equation}
where $\mathbf{P} \in \mathbb{R}^{k \times d}$ is a sparse matrix with the nonzeros entries generated iid from the standard distribution. $\mathbf{H} \in \mathbb{R}^{d \times d}$ is the Hadamard matrix, and $\mathbf{D} \in \mathbb{R}^{d \times d}$ is a diagonal matrix with random signs. 
Although the Hadamard transformation has computational complexity $\mathcal{O}(d \log d)$ (the multiplication with $\mathbf{H}$), this method is often slower than CBE due to the sparse Gaussian projection step (i.e., multiplication by $\mathbf{P}$). 

The second method we compare with is Fast Binary Embedding (FBE). It is a theoretically sound method recently proposed in \cite{yi2015binary}. 
FBE generates binary bits using a partial Walsh-Hadamard matrix and a set of partial Gaussian Toeplitz matrices. 
The method can achieve the optimal measurement complexity $\mathcal{O}(\frac{1}{\delta^2} \log N)$. 
We follow the parameters settings in \cite{yi2015binary} (Flickr-25600 dataset, with the number of bits 5,000, 10,000 and 15,000). Note that under the setting, FBE is at least a few times slower than CBE due to the use of multiple Toeplitz projections. 
Figure \ref{fig:cbe_vs_fbe} shows the retrieval performance.
Based on the experiments, in addition to being much faster than FBE and FJLT, CBE-rand provides comparable or even better performance. 
Another advantage of CBE is that the framework permits data-dependent optimization to further improve the performance.  In all the experiments, CBE-opt achieves the highest recall by a large margin compared to other methods. 

\begin{table}
\center
\begin{tabular}{c|c|c|c}
\hline
  Original& LSH  & Bilinear-opt  & CBE-opt  \\ 
\hline  25.59$\pm$0.33 & 23.49$\pm$0.24  & 24.02$\pm$0.35 & 24.55 $\pm$0.30  \\ 
\hline
\end{tabular}
\caption{Multiclass classification accuracy (\%) on binary coded ImageNet-25600. The binary codes of same dimensionality are 32 times more space efficient than the original features (single-float).}
\label{table:classification}
\end{table}

\subsection{Classification} Besides retrieval, we also test the binary codes for classification. The advantage is to save on storage, allowing even large scale datasets to fit in memory  \cite{li2011hashing, sanchez2011high}.
We follow the asymmetric setting of \cite{sanchez2011high} by training linear SVM on binary code $\text{sign}(\mathbf{R} \mathbf{x})$, and testing on the original $\mathbf{R} \mathbf{x}$. Empirically, this has been shown to give better accuracy than the symmetric procedure.
We use ImageNet-25600, with randomly sampled 100 images per category for training, 50 for validation and 50 for testing. The code dimension is set as 25,600.
As shown in Table \ref{table:classification}, CBE, which has much faster computation, does not show any performance degradation compared with LSH or bilinear codes in classification task.

\section{Conclusion}
We proposed a method of binary embedding for high-dimensional data.
Central to our framework is to use a type of highly structured matrix,
the circulant matrix, to perform the linear projection.  The proposed
method has time complexity $\mathcal{O}(d\log d)$ and space complexity
$\mathcal{O}(d)$, while showing no performance degradation on
real-world data compared with more expensive approaches
($\mathcal{O}(d^2)$ or $\mathcal{O}(d^{1.5})$). The parameters of the
method can be randomly generated, where interesting theoretical
analysis was carried out to show that the angle preserving quality can
be as good as LSH.  The parameters can also be learned based on
training data with an efficient optimization algorithm. 

\bibliography{cbe_jmlr}

\newcommand{\etalchar}[1]{$^{#1}$}
\begin{thebibliography}{CYF{\etalchar{+}}15b}

\bibitem[AC06]{ailon2006approximate}
Nir Ailon and Bernard Chazelle.
\newblock Approximate nearest neighbors and the fast {J}ohnson-{L}indenstrauss
  transform.
\newblock In {\em Proceedings of the ACM Symposium on Theory of Computing},
  2006.

\bibitem[Ach03]{achlioptas2003database}
Dimitris Achlioptas.
\newblock Database-friendly random projections: {J}ohnson-{L}indenstrauss with
  binary coins.
\newblock {\em Journal of Computer and System Sciences}, 2003.

\bibitem[Cha02]{charikar2002similarity}
Moses~S Charikar.
\newblock Similarity estimation techniques from rounding algorithms.
\newblock In {\em Proceedings of the ACM Symposium on Theory of Computing},
  2002.

\bibitem[CKB{\etalchar{+}}15]{Choromanska15}
Anna Choromanska, Choromanski Krzysztof, Mariusz Bojarski, Tony Jebara, Sanjiv
  Kumar, and Yann LeCun.
\newblock Binary embeddings with structured hashed projections.
\newblock {\em arXiv preprint arXiv:1511.05212v1}, 2015.

\bibitem[CYF{\etalchar{+}}15a]{iccv15_circulant}
Yu~Chen, Felix~Xinnan Yu, Rogerio Feris, Sanjiv Kumar, and S.-F. Choudhary,
  Alok abd~Chang.
\newblock An exploration of parameter redundancy in deep networks with
  circulant projections.
\newblock In {\em Proceedings of the IEEE International Conference on Computer
  Vision}, 2015.

\bibitem[CYF{\etalchar{+}}15b]{arxiv_circulant_nn}
Y.~Cheng, Felix~Xinnan Yu, R.S Feris, S.~Kumar, A.~Choudhary, and S.-F. Chang.
\newblock Fast neural networks with circulant projections.
\newblock {\em arXiv preprint arXiv:1502.03436}, 2015.

\bibitem[DDS{\etalchar{+}}09]{deng2009imagenet}
Jia Deng, Wei Dong, Richard Socher, Li-Jia Li, Kai Li, and Li~Fei-Fei.
\newblock Imagenet: A large-scale hierarchical image database.
\newblock In {\em Proceedings of the IEEE Conference on Computer Vision and
  Pattern Recognition}, 2009.

\bibitem[DKS11]{dasgupta2011fast}
Anirban Dasgupta, Ravi Kumar, and Tam{\'a}s Sarl{\'o}s.
\newblock Fast locality-sensitive hashing.
\newblock In {\em Proceedings of the ACM SIGKDD Conference on Knowledge
  Discovery and Data Mining}, 2011.

\bibitem[FM88]{frankl1988johnson}
Peter Frankl and Hiroshi Maehara.
\newblock The johnson-lindenstrauss lemma and the sphericity of some graphs.
\newblock {\em Journal of Combinatorial Theory, Series B}, 44(3):355--362,
  1988.

\bibitem[GKRL13]{gonglearning}
Yunchao Gong, Sanjiv Kumar, Henry~A Rowley, and Svetlana Lazebnik.
\newblock Learning binary codes for high-dimensional data using bilinear
  projections.
\newblock In {\em Proceedings of the IEEE Conference on Computer Vision and
  Pattern Recognition}, 2013.

\bibitem[GKVL12]{gong2012angular}
Yunchao Gong, Sanjiv Kumar, Vishal Verma, and Svetlana Lazebnik.
\newblock Angular quantization-based binary codes for fast similarity search.
\newblock In {\em Advances in Neural Information Processing Systems}, 2012.

\bibitem[GLGP12]{gongiterative}
Y.~Gong, S.~Lazebnik, A.~Gordo, and F.~Perronnin.
\newblock Iterative quantization: A procrustean approach to learning binary
  codes for large-scale image retrieval.
\newblock {\em IEEE Transactions on Pattern Analysis and Machine Intelligence},
  PP(99):1, 2012.

\bibitem[GP11]{gordo2011asymmetric}
Albert Gordo and Florent Perronnin.
\newblock Asymmetric distances for binary embeddings.
\newblock In {\em Proceedings of the IEEE Conference on Computer Vision and
  Pattern Recognition}, 2011.

\bibitem[Gra06]{gray2006toeplitz}
Robert~M Gray.
\newblock {\em Toeplitz and circulant matrices: A review}.
\newblock Now Pub, 2006.

\bibitem[HV11]{hinrichs2011johnson}
Aicke Hinrichs and Jan Vyb{\'\i}ral.
\newblock Johnson-{L}indenstrauss lemma for circulant matrices.
\newblock {\em Random Structures \& Algorithms}, 39(3):391--398, 2011.

\bibitem[IM98]{indyk1998approximate}
Piotr Indyk and Rajeev Motwani.
\newblock Approximate nearest neighbors: towards removing the curse of
  dimensionality.
\newblock In {\em Proceedings of the ACM Symposium on Theory of Computing},
  1998.

\bibitem[JDS11]{jegou2011product}
Herve Jegou, Matthijs Douze, and Cordelia Schmid.
\newblock Product quantization for nearest neighbor search.
\newblock {\em IEEE Transactions on Pattern Analysis and Machine Intelligence},
  33(1):117--128, 2011.

\bibitem[JL84]{johnson1984extensions}
William~B Johnson and Joram Lindenstrauss.
\newblock Extensions of lipschitz mappings into a hilbert space.
\newblock {\em Contemporary Mathematics}, 26(189-206):1, 1984.

\bibitem[KD09]{Kulislearningto}
Brian Kulis and Trevor Darrell.
\newblock Learning to hash with binary reconstructive embeddings.
\newblock In {\em Advances in Neural Information Processing Systems}, 2009.

\bibitem[LAS08]{liberty2008dense}
Edo Liberty, Nir Ailon, and Amit Singer.
\newblock Dense fast random projections and lean walsh transforms.
\newblock {\em Approximation, Randomization and Combinatorial Optimization.
  Algorithms and Techniques}, pages 512--522, 2008.

\bibitem[LSMK11]{li2011hashing}
Ping Li, Anshumali Shrivastava, Joshua Moore, and Arnd~Christian Konig.
\newblock Hashing algorithms for large-scale learning.
\newblock In {\em Advances in Neural Information Processing Systems}, 2011.

\bibitem[LSS13]{le2013fastfood}
Quoc Le, Tam{\'a}s Sarl{\'o}s, and Alex Smola.
\newblock Fastfood -- approximating kernel expansions in loglinear time.
\newblock In {\em Proceedings of the International Conference on Machine
  Learning}, 2013.

\bibitem[LWKC11]{liu2011hashing}
Wei Liu, Jun Wang, Sanjiv Kumar, and Shih-Fu Chang.
\newblock Hashing with graphs.
\newblock In {\em Proceedings of the International Conference on Machine
  Learning}, 2011.

\bibitem[Mat08]{matouvsek2008variants}
Ji{\v{r}}{\'\i} Matou{\v{s}}ek.
\newblock On variants of the {Johnson--Lindenstrauss} lemma.
\newblock {\em Random Structures \& Algorithms}, 33(2):142--156, 2008.

\bibitem[NF12]{Norouzi11}
Mohammad Norouzi and David Fleet.
\newblock Minimal loss hashing for compact binary codes.
\newblock In {\em Proceedings of the International Conference on Machine
  Learning}, 2012.

\bibitem[NFS12]{norouzi2012Hamming}
Mohammad Norouzi, David Fleet, and Ruslan Salakhutdinov.
\newblock Hamming distance metric learning.
\newblock In {\em Advances in Neural Information Processing Systems}, 2012.

\bibitem[OSB{\etalchar{+}}99]{oppenheim1999discrete}
Alan~V Oppenheim, Ronald~W Schafer, John~R Buck, et~al.
\newblock {\em Discrete-time signal processing}, volume~5.
\newblock Prentice Hall Upper Saddle River, 1999.

\bibitem[RL09]{raginsky2009locality}
Maxim Raginsky and Svetlana Lazebnik.
\newblock Locality-sensitive binary codes from shift-invariant kernels.
\newblock In {\em Advances in Neural Information Processing Systems}, 2009.

\bibitem[RV13]{rudelson2013hanson}
Mark Rudelson and Roman Vershynin.
\newblock Hanson-wright inequality and sub-gaussian concentration.
\newblock {\em Electron. Commun. Probab}, 18(0), 2013.

\bibitem[SP11]{sanchez2011high}
Jorge S\'anchez and Florent Perronnin.
\newblock High-dimensional signature compression for large-scale image
  classification.
\newblock In {\em Proceedings of the IEEE Conference on Computer Vision and
  Pattern Recognition}, 2011.

\bibitem[Vyb11]{vybiral2011variant}
Jan Vyb{\'\i}ral.
\newblock A variant of the {J}ohnson--{L}indenstrauss lemma for circulant
  matrices.
\newblock {\em Journal of Functional Analysis}, 260(4):1096--1105, 2011.

\bibitem[WKC10]{wang2010sequential}
Jun Wang, Sanjiv Kumar, and Shih-Fu Chang.
\newblock Sequential projection learning for hashing with compact codes.
\newblock In {\em Proceedings of the International Conference on Machine
  Learning}, 2010.

\bibitem[WTF08]{weiss2008spectral}
Yair Weiss, Antonio Torralba, and Rob Fergus.
\newblock Spectral hashing.
\newblock In {\em Advances in Neural Information Processing Systems}, 2008.

\bibitem[YCP15]{yi2015binary}
Xinyang Yi, Constantine Caramanis, and Eric Price.
\newblock Binary embedding: Fundamental limits and fast algorithm.
\newblock {\em arXiv preprint arXiv:1502.05746}, 2015.

\bibitem[YKGC14]{icml14_cbe}
Felix~Xinnan Yu, S.~Kumar, Y.~Gong, and S.-F. Chang.
\newblock Circulant binary embedding.
\newblock In {\em Proceedings of the International Conference on Machine
  Learning}, 2014.

\bibitem[YKRC15]{arxiv_cnm}
Felix~Xinnan Yu, Sanjiv Kumar, Henry Rowley, and Shih-Fu Chang.
\newblock Compact nonlinear maps and circulant extensions.
\newblock {\em arXiv preprint arXiv:1503.03893}, 2015.

\bibitem[YMD{\etalchar{+}}14]{yang2014deep}
Zichao Yang, Marcin Moczulski, Misha Denil, Nando de~Freitas, Alex Smola,
  Le~Song, and Ziyu Wang.
\newblock Deep fried convnets.
\newblock {\em arXiv preprint arXiv:1412.7149}, 2014.

\bibitem[ZC13]{zhang2013new}
Hui Zhang and Lizhi Cheng.
\newblock New bounds for circulant {J}ohnson-{L}indenstrauss embeddings.
\newblock {\em arXiv preprint arXiv:1308.6339}, 2013.

\end{thebibliography}

\appendix
\section{Proofs of the Technical Lemmas}\label{sec:app}
\subsection{Proof of Lemma \ref{lem:error-bound}}\label{app:lem:error-bound}
For convenience, define $\bu^\perp = \bu - \Pi \bu$, and similarly define $\bv^\perp$.  From our earlier observation about independence, we have that
\begin{equation}
\mathbb{E} \left[ \left( \frac{1 - \sign(\br^T \ba ) \sign(\br^T \bb ) }{2} - \frac{\theta}{\pi}  \right)
 \left( \frac{1 - \sign(\br^T \bu^\perp ) \sign(\br^T \bv^\perp ) }{2} - \frac{\theta}{\pi}  \right) \right] = 0.
 \end{equation}
Because the LHS is equal to the product of the expectations, and the first term is $0$.  Thus the quantity we wish to bound is 
\[ \mathbb{E} \left[ \left( \frac{1 - \sign(\br^T \ba ) \sign(\br^T \bb ) }{2} - \frac{\theta}{\pi}  \right)
 \left( \frac{\sign(\br^T \bu) \sign(\br^T \bv) - \sign(\br^T \bu^\perp ) \sign(\br^T \bv^\perp ) }{2}  \right) \right]. \nonumber \]
Now by using the fact that $\mathbb{E} [XY] \le \mathbb{E} [|X||Y|]$, together with the observation that the quantity $|(1- \sign(\br^T \ba ) \sign(\br^T \bb ) )/2 - \theta/\pi |$ is at most $2$,
we can bound the above by
\[ \mathbb{E} \left[ | \sign(\br^T \bu) \sign(\br^T \bv) - \sign(\br^T \bu^\perp ) \sign(\br^T \bv^\perp) |  \right]. \]
This is equal to 
\[ 2 \Pr[ \sign(\br^T \bu) \sign(\br^T \bv) \ne \sign(\br^T \bu^\perp ) \sign(\br^T \bv^\perp)] ,\]
since the term in the expectation is $2$ if the product of signs is different, and $0$ otherwise.
To bound this, we first observe that for any two unit vectors $\bx, \by$ with $\angle(\bx, \by) \le \epsilon$, we have $\Pr[ \sign(\br^T \bx ) \ne \sign(\br^T \by)] \le \epsilon/\pi$. We can use this to say that
\[ \Pr[  \sign(\br^T \bu) \ne \sign(\br^T \bu^\perp) ] = \frac{ \angle (\bu, \bu^\perp) }{\pi}. \]
This angle can be bounded in our case by $(\pi/2) \cdot \delta$ by basic geometry.\footnote{$\bu$ is a unit vector, and $\bu^\perp + \Pi \bu = \bu$, and $\norm{\Pi \bu} \le \delta$, so the angle is at most $\sin^{-1} (\delta)$.}  Thus by a union bound, we have that 
\[ \Pr[ \big( \sign(\br^T \bu) \ne \sign(\br^T \bu^\perp)\big) \vee \big( \sign(\br^T \bv) \ne \sign(\br^T \bv^\perp)\big)  ]  \le \delta.\]
This completes the proof.

\subsection{Proof of Lemma \ref{lem:small-dproduct}}\label{app:lem:small-dproduct}
Denoting the $i$th entry of $\bp$ by $p_i$ (so also for $\bq$), we have that
\[ S := \iprod{ \bD\bp, s_{\rightarrow t}(\bD\bq) } = \sum_{i=0}^{d-1} \sigma_i \sigma_{i+t} p_i q_{i+t}.\]
We note that $\E[ S ] = 0$, by linearity of expectation (as $t >0$, $\E[ \sigma_i \sigma_{i+t} ] =0$), thus the lemma is essentially a tail bound on $S$.  While we can appeal to standard tail bounds for quadratic forms of sub-Gaussian random variables (e.g. Hansen-Wright~\cite{rudelson2013hanson}), we give below a simple argument.
Let us define 
\[ f( \sigma_0, \sigma_1, \dots, \sigma_{d-1} ) = \sum_{i=0}^{d-1} p_i q_{i+t} \sigma_i \sigma_{i+t}.\]
We will view $f$ as being obtained from a martingale as follows. Define
\[ Q_i := f(\sigma_0, \sigma_1, \dots, \sigma_i, 0, \dots, 0) - f(\sigma_0, \sigma_1, \dots, \sigma_{i-1}, 0, \dots, 0). \]
In this notation, we have $S = Q_0 + Q_1 + \dots + Q_{d-1}$. 

We have the martingale property that $\E[ Q_i | Q_0, Q_1, \dots, Q_{i-1}] = 0$ for all $i$, (because $\sigma_i$ is $\pm 1$ with equal probability). Further, we have the {\em bounded difference} property, i.e., $|Q_i| \le |p_i q_{i+t}| + |p_{i-t} q_i|$. This  implies that
\[ |Q_i|^2 \le 2 (p_i^2 q_{i+t}^2 + p_{i-t}^2 q_i^2). \]
Thus we can use Azuma's inequality to conclude that for any $\gamma >0$,
\[ \Pr[ | \sum_i Q_i - \E[ \sum_i Q_i ] | > \gamma ] < e^{ - \frac{\gamma^2}{2 \cdot \sum_i 2( p_i^2 q_{i+t}^2 + p_{i-t}^2 q_{i}^2 )} } = e^{- \frac{\gamma^2}{8\sum_i^2 p_i^2 q_{i+t}^2 }}.\]
We can now use the fact that $\sum_i p_i^2 q_{i+t}^2 \le \rho^2 \sum_i p_i^2 = \rho^2$ (since $\norm{\bp}_2 = 1$ and $\norm{\bq}_\infty \le \rho$).  This establishes the lemma.
%

\subsection{Proof of Lemma \ref{lem:jl-orthogonality}}\label{app:lem:jl-orthogonality}

First, using Lemma~\ref{lem:small-dproduct} we have, for any $i \ne j$ and $c>0$,
\[ \Pr[ |\iprod{ X_i, Y_j }| > c ] < e^{-c^2/8\rho^2}. \]
We have a similar bound for $\Pr[ |\iprod{ X_i, X_j}| > c]$.  Thus by setting $c = 4 \rho \sqrt{\log (k/\delta)}$ ($\delta$ as in the statement of the lemma), we can take a union bound over all $k^2$ choices of $i \ne j$ and conclude that w.p. at least $1-\delta$, we have
\begin{equation}\label{eq:max-correlation}
\max_{i \ne j} \{ |\iprod{ X_i, X_j}|, |\iprod{X_i, Y_j}| \} < 4 \rho \sqrt{\log(k/\delta)}.
\end{equation}
%
%
%
We now prove that whenever Eq.~\eqref{eq:max-correlation} holds, we obtain $(\gamma, k)$ orthogonality for the desired $\gamma$.  Let us start with a basic fact in linear algebra.

\begin{lemma}\label{claim:small-coeffs}
Let $A$ be an $d \times k$ matrix with $\sigma_k (A) \ge \tau$, for some parameter $\tau$.  Then any unit vector in the column span of $A$ can be written as $\sum_i \alpha_i A_i$, with $\sum_i \alpha_i^2 \le 1/\tau^2$.
\end{lemma}
\begin{proof}
By the definition of $\sigma_k$, we have that for any $\alpha_i$, $\norm{\sum_i \alpha_i A_i}_2^2 \ge \tau^2 \left( \sum_i \alpha_i^2 \right)$.  Thus for any unit vector $\sum_i \alpha_i A_i$, we have $\sum_i \alpha_i^2 \le 1/\tau^2$.
\end{proof}

Now let $B$ be the $d \times 2k$ matrix whose columns are $X_1, Y_1, X_2, Y_2, \dots, X_k, Y_k$ in that order.  Consider the entries of $B^T B$.  Since $X_i, Y_i$ are unit vectors, the diagonals are all $1$.  The $(2i-1, 2i)$th and $(2i, 2i-1)$th entries are exactly $\cos \theta$, because the angle between $X_i, Y_i$ is $\theta$.  The rest of the entries are of magnitude $<\eta := 4\rho \sqrt{\log (k/\delta)}$.

Thus if we consider $M = B^T B - I$ (diagonal removed from $B^T B$), we have $-(\cos\theta + k\eta)I \preceq M \preceq (\cos \theta + k\eta)I $ (diagonal dominance).  Thus we conclude that $B^T B$ has all its eigenvalues $\ge 1-\cos \theta-k\eta$.  Since $\theta \in (0, \pi/2)$, we can use the standard inequality $\cos \theta < 1-\theta^2/2$ to conclude that the eigenvalues are $\ge \theta^2/2 -k\eta$.  Now by our assumption on $\rho$, we have that $k \eta < \theta^2/4$.  This implies that all the eigenvalues are $\ge \theta^2/4$.

Thus we have $\sigma_{2k}^2 (B) \ge \theta^2/4$.  We prove now that this lets us obtain a decomposition that helps us prove $(\gamma, k)$-orthogonality. A crucial observation is the following.
\begin{lemma}\label{claim:proj-length}
The projection of $X_i$ onto $\spn{X_1, Y_1, X_2, Y_2, \dots, X_{i-1}, Y_{i-1}}$ has length at most $\frac{2\eta \sqrt{2k}}{\theta}$.  
\end{lemma}
\begin{proof}
Let $\calS$ denote $\spn{X_1, Y_1, \dots, X_{i-1}, Y_{i-1}}$.  By definition, the squared of the length of projection is equal to $\max\{ \iprod{y, X_i}^2 ~|~ y \in \calS \text{ and }  \norm{y}_2=1 \}$ (this is how the projection onto a subspace can be defined).

To bound this, consider any unit vector $y \in \calS$, and suppose we write it as $\sum_{j < i} \alpha_j X_j + \beta_j Y_j$.  Let $B'$ be the matrix that has columns $X_j, Y_j$, $j < i$.  Then it is straightforward to see that $\sigma_{2(i-1)}(B') \ge \sigma_{2k} (B) \ge \theta/2$.  Thus Claim~\ref{claim:small-coeffs} implies that $\sum_{j <i} \alpha_j^2 + \beta_j^2 \le 4/\theta^2$.  This means that
\begin{align}
\iprod{X_i, y}^2 &= \big( \sum_{j < i} \alpha_j \iprod{X_j, X_i} + \beta_j \iprod{Y_j, X_i} \big)^2 \\
&\le \big( \sum_{j < i} \alpha_j^2 + \beta_j^2 \big) \big( \sum_{j < i} \iprod{X_j, X_i}^2 + \iprod{Y_j, X_i}^2  \big) \\
&\le \frac{4}{\theta^2} \cdot (2i-2)\eta^2.
\end{align}
(In the first step, we used Cauchy-Schwartz.) Taking square roots now gives the claim.
\end{proof}

Now we perform the following procedure on the vectors (it is essentially Gram-Schmidt orthonormalization, with the slight twist that we deal with $X_i, Y_i$ together):
\begin{enumerate}
\item Initialize:  $\bu_1 = X_1, \quad  \be_1 = 0, \quad \bv_1 = Y_1, \quad \Bf_1 = 0.$
\item For $i = 2, \dots, k$, we set $\bu_i, \bv_i$ to be the projections of $X_i, Y_i$ (respectively) orthogonal to $\spn{X_1, Y_1, \dots, X_{i-1}, Y_{i-1}}$.  Set $\be_i = X_i - \bu_i$ and $\Bf_i = Y_i - \bv_i$. 
\end{enumerate}

The important observation is that for any $i$, we have
\[ \spn{\bu_j, \bv_j : j<i } = \spn{\bu_j, \bv_j, \be_j, \Bf_j: j<i} = \spn{X_j, Y_j:j<i}.\]
This is because by definition, $\be_i, \Bf_i \in \spn{X_j, Y_j : j<i}$ for all $i$.  Thus we have that $\bu_i$ and $\bv_i$ satisfy the first condition in Definition~\ref{def:orthogonality}.  It just remains to analyze the lengths.  Now we can use Claim~\ref{claim:proj-length} to conclude that 
\[ \norm{\be_i}_2^2, \norm{\Bf_i}_2^2 < \frac{8k\eta^2}{\theta^2}  = \frac{128 \cdot k\rho^2 \log (k/\delta)}{\theta^2}. \]
Once again, we use the bound on $\rho$ to conclude that this quantity is at most $16\rho$.  This  completes the proof of  Lemma~\ref{lem:jl-orthogonality}, with $\gamma = 4\sqrt{\rho}$.  \qed

\subsection{Proof of Lemma \ref{lem:orthog-to-thm}}\label{app:lem:orthog-to-thm}
We start with a simple claim about the angle between $\bu_i$ and $\bv_i$. 
\begin{lemma}\label{claim:angle-bound}
For all $i$, we have $\angle (\bu_i, \bv_i) \in (\theta - \pi\gamma, \theta + \pi\gamma)$.
\end{lemma}
\begin{proof}
The angle between $X_i$ and $\bu_i$ is at most $\sin^{-1} (\gamma) < (\pi/2)\gamma$. So also, the angle between $Y_i$ and $\bv_i$ is at most $(\pi/2)\gamma$.  Thus the angle between $\bu_i, \bv_i$ is in the interval $(\theta - \pi\gamma, \theta+\pi\gamma)$ (by triangle inequality for the geodesic distance).
\end{proof}

Let $\eta > 0$ be a parameter we will fix later (it will be a constant times $\gamma \sqrt{\log (k/\delta)}$). For all $i$, we define the following events:
\begin{align}
E_i &: ~~ \min \{ \iprod{ \br, \bu_i }, \iprod{\br, \bv_i} \} < \eta \\
F_i &: ~~  \neg E_i \text{ and } \sign{ \iprod{\br, \bu_i} } \ne \sign{ \iprod{\br, \bv_i} }
\end{align}

The following claim now follows easily.
\begin{lemma}
For any $i$, we have
\begin{align}
\Pr[ E_i ] &\le 2\eta \label{eq:prob-ei},\\
\Pr[ F_i] &\in \left( \frac{\theta}{\pi} - \pi \gamma - 2\eta, \frac{\theta}{\pi} \right) \label{eq:prob-fi}
\end{align}
\end{lemma}
\begin{proof}
The first inequality follows from the small ball probability of a univariate Gaussian (since $\iprod{\br, \bu_i}$ is a Gaussian of unit variance), and the second follows from Claim~\ref{claim:angle-bound} and \eqref{eq:prob-ei}.
\end{proof}

We will set $\eta$ to be larger than $\pi \gamma$, so the RHS in~\eqref{eq:prob-fi} can be replaced with $(\nicefrac{\theta}{\pi} - 3\eta, \nicefrac{\theta}{\pi})$.  Furthermore, the events above for a given $i$ depend {\em only} on the projection of $\br$ to $\spn{\bu_i, \bv_i}$; thus they are independent for different $i$.  Let us abuse notation slightly and denote by $E_i$ also the indicator random variable for the event $E_i$ (so also $F_i$).  Then by standard Chernoff bounds, we have for any $\tau >0$,
\begin{align}
\Pr \left[ \sum_i E_i \ge 2k\eta + k\tau \right] < e^{-\frac{k\tau^2}{4\eta + \tau}}  \label{eq:prob-ei-sum}, \\
\Pr \left[ \sum_i F_i \not\in \left( \frac{k\theta}{\pi} - 3k\eta - k\tau, \frac{k\theta}{\pi} + k\tau \right) \right] < 2e^{-\frac{k\tau^2}{\theta + \tau} } \label{eq:prob-fi-sum}.
\end{align}
Finally let $H$ denote the event: 
\[ \max_i \{ |\iprod{\br, \be_i}|, |\iprod{\br, \Bf_i}| \} \ge \eta. \]
For any $i$, since $\norm{\be_i} < \gamma$, we have $\Pr[ |\iprod{\br, \be_i} | > t\gamma] \le e^{-t^2/2}$.  We can use the same bound with $\Bf_i$, and take a union bound over all $i$, to conclude that $\Pr[ H ] \le 2k \cdot e^{-\eta^2/2\gamma^2}$. 

Let us call a choice of $\br$ {\em good} if neither of the events in \eqref{eq:prob-ei-sum}-\eqref{eq:prob-fi-sum}  above occur, and additionally $H$ does not occur.  Clearly, the probability of an $\br$ being good is at least $1- \delta$, provided $\tau$ and $\eta$ are chosen such that the RHS of the tail bounds above are all made $\le \delta/4$.  

Before setting these values, we note  that for a good $\br$,
\[ \frac{1}{k} \sum_{i} \mathbf{1}\{ \sign{\iprod{\br, X_i}} \ne \sign{\iprod{\br, Y_i}} \} \in \left( \frac{\theta}{\pi} - 3\eta - \tau, \frac{\theta}{\pi} + 2\eta + 2\tau \right). \]
This is because whenever $F_i \wedge \neg H$ occurs, we have $\sign{\iprod{\br, X_i}} \ne \sign{\iprod{\br, Y_i}}$, and thus the LHS above is at least $\frac{\theta}{\pi} - 3\eta - \tau$. Also if we have $\neg H$, then the only way we can have $\sign{\iprod{\br, X_i}} \ne \sign{\iprod{\br, Y_i}}$ is if either $F_i$ occurs, or if $E_i$ occurs (in the latter case, it is not necessary that the signs are unequal).  Thus we can upper bound the LHS by $\frac{\theta}{\pi} +2\eta + 2\tau$.

Let us now set the values of $\eta$ and $\tau$.  From the above, we need to ensure:
\begin{equation}
\frac{k\tau^2}{4\eta + \tau} \ge \log (4/\delta), \quad  \frac{k\tau^2}{\theta+\tau} \ge \log(8/\delta), \quad \text{and} \quad \frac{\eta^2}{2\gamma^2} \ge \log (4k/\delta).
\end{equation}
Thus we set $\eta = 2\gamma \sqrt{ \log (4k/\delta)}$, and 
\[ \tau \ge \max \left\{ \frac{2\log (8/\delta)}{k},  \sqrt{\frac{2 \theta \log(8/\delta)}{k}} , \sqrt{ \frac{8\eta \log(4/\delta)}{k} }  \right\}. \]
For the above inequality to hold, it suffices to set
\[ \tau \ge \frac{8\log(1/\delta)}{\sqrt{k}}. \] 
This gives the desired bound on the deviation in the angle.


\end{document}